\documentclass{amsart}
\usepackage{soul}
\usepackage[title]{appendix}
\usepackage{graphicx}
\usepackage{amsfonts}
\usepackage{url}
\usepackage{amscd}
\usepackage{amssymb}

\usepackage{algorithm2e}
\usepackage[hidelinks]{hyperref}

\newtheorem{remark}{Remark}[section]

\newtheorem{theorem}{Theorem}[section]

\newtheorem{definition}[theorem]{Definition}
\newtheorem{proposition}[theorem]{Proposition}
\numberwithin{equation}{section}
\def\L{\mathcal{L}}
\usepackage{indentfirst}
\usepackage[greek,english]{babel}
\usepackage[utf8]{inputenc}
\usepackage{color,amsmath,amssymb,graphicx}
\usepackage{tikz-cd}
\usepackage{amsmath,amssymb}
\usepackage{enumerate}
\usepackage{dsfont}
\usepackage{amsthm}
\usepackage{float}
\usepackage{mathtools}

\usepackage{listings}
\definecolor{codegreen}{rgb}{0,0.6,0}
\definecolor{codegray}{rgb}{0.5,0.5,0.5}
\definecolor{codepurple}{rgb}{0.58,0,0.82}
\lstdefinestyle{mystyle}{
    backgroundcolor=\color{white},   
    commentstyle=\color{codegreen},
    keywordstyle=\color{magenta},
    numberstyle=\tiny\color{codegray},
    stringstyle=\color{codepurple},
    basicstyle=\ttfamily\footnotesize,
    breakatwhitespace=false,         
    breaklines=true,                 
    captionpos=b,                    
    keepspaces=true,                 
    numbers=left,                    
    numbersep=5pt,                  
    showspaces=false,                
    showstringspaces=false,
    showtabs=false,                  
    tabsize=2,
    upquote=true,
    columns=fullflexible
}

\author[E. Poimenidou]{Eirini Poimenidou}
\address{E. Poimenidou, School of Informatics, Aristotle University of Thessaloniki, Greece \& CERN, Geneva, Switzerland}
\email{epoimeni@csd.auth.gr, epoimeni@cern.ch}
\author[M. Adamoudis]{Marios Adamoudis}
\address{M. Adamoudis, }
\email{marios.p7@hotmail.com} 
 \author[K. A. Draziotis]{Konstantinos A. Draziotis}
\address{K. A. Draziotis, School of Informatics, Aristotle University of Thessaloniki, Greece}
\email{drazioti@csd.auth.gr}
 \author[K. Tsichlas]{Kostas Tsichlas}
 \address{K. Tsichlas, Department of Computer Engineering and Informatics, University of Patras, Greece}
\email{ktsichlas@ceid.upatras.gr}

\begin{document}
	\title{Message Recovery Attack in NTRU through VFK Lattices}
\maketitle	
	\begin{abstract}

 In the present paper, we implement a message recovery attack to all variants of the NTRU cryptosystem.
 Our approach involves a reduction from the NTRU-lattice to a Voronoi First Kind lattice, enabling the application of a polynomial CVP exact algorithm crucial for executing the Message Recovery. The efficacy of our attack relies on a specific oracle that permits us to approximate an unknown quantity.  Furthermore, we outline the mathematical conditions under which the attack is successful. Finally, we delve into a well-established polynomial algorithm for CVP on VFK lattices and its implementation, shedding light on its efficacy in our attack.  Subsequently, we present comprehensive experimental results on the NTRU-HPS and the NTRU-Prime variants of the NIST submissions and propose a method that could indicate the resistance of the NTRU cryptosystem to our attack.  
 	\end{abstract}

\section{Introduction}
In 1996, Hoffstein, Pipher, and Silverman developed the NTRU cryptosystem, aiming to create robust encryption and signature systems, as detailed in \cite{hoffstein}.
 Its security is based on the difficulty of solving a system of linear equations over polynomial rings, a problem that is expected to remain hard even with quantum computers. While it is not equivalent to the Shortest and Closest Vector Problems (SVP/CVP) in lattices, the NTRU cryptosystem has withstood over 25 years of cryptanalysis, and variants of it have been shown to be secure under the Ring Learning With Error (R-LWE) hardness assumption. NTRU is known for its exceptional performance and moderate key-size, making it a popular choice for embedded cryptography. It has been standardized by IEEE, X9.98, and PQCRYPTO, and was a finalist in the NIST post-quantum cryptography standardization effort. Ultimately, the Kyber algorithm was selected for standardization. It is also noted in the 3rd Status Report of NIST for PQC\footnote{\url{https://nvlpubs.nist.gov/nistpubs/ir/2022/NIST.IR.8413.pdf}} that the overall performance of any of these KEMs (i.e. Kyber, Saber and NTRU) would be acceptable for general-use applications. In addition, NTRU-HRSS was chosen by Google for their ALTS protocol (ALTS : Application Layer Transport Security) in the hybrid model\footnote{the url link is \href{https://cloud.google.com/blog/products/identity-security/why-google-now-uses-post-quantum-cryptography-for-internal-comms}{here}}, and NTRU-Prime \cite{ ntru-prime-submission-3} was implemented in openssh 9.0 in a hybrid scheme with EC-Diffie-Hellman, see \cite{openssh}.

In this paper, we present a message recovery attack applicable to all NTRU variants, assuming a specific oracle. We detail the oracle's function and outline the mathematical conditions under which our attack proves successful. We reduce the message recovery problem to a CVP problem over a Voronoi First Kind (VFK) lattice, where a polynomial algorithm exists both for Shortest Vector and Closest Vector problems. To the best of our knowledge, the association between the NTRU problem and VFK lattices does not exist in the current bibliography.

We continue with an exposition of NTRU-HPS, the initial variant of NTRU, followed by our contribution and a review of prior research on the topic.

\subsection{NTRU-HPS cryptosystem}\label{sec:ntru}
Let the polynomial ring ${\mathcal{R}}={\mathbb{Z}}[x]/\langle D(x) \rangle$ for some $D(x)\in {\mathbb{Z}}[x]$ and $\langle D(x) \rangle$ the ideal generated by $D(x).$ We write $*$ for the multiplication in the ring. Also, fix a polynomial $h(x)\in {\mathbb{Z}}[x]$ of degree $N-1.$ We set, 
\begin{equation}\label{B_h}
B_h = \begin{pmatrix}
-\ h(x)\ - \\
-\ x*h(x)\ - \\
\vdots\\
-\ x^{N-1}*h(x)\ -
\end{pmatrix},
\end{equation}
where with $x^{i}*h(x),$ we mean the vector with coordinates the coefficients of the polynomial $h(x),$ after multiplication in ${\mathcal{R}}$ with $x^i.$ In expressing the coefficient vector of $h(x)=a_{N-1}x^{N-1}+\cdots + a_0,$ we denoted as ${\textbf h}=(a_0,...,a_{N-1}).$
Then, the multiplication $g(x)*h(x)$ in ${\mathcal{R}}$ can be represented as the multiplication of the row matrix $[{\textbf g}]$ and matrix $B_h,$ i.e., $[{\textbf g}]B_h.$ 

The set 
$${\mathcal{L}}_h=\{(f(x),g(x))\in {\mathcal{R}}^2 : g(x)=f(x)*h(x)\},$$ where $h(x)$ has degree say $N-1$ is a lattice. Indeed,
$${\mathcal{L}}_h={\mathbb{Z}}^{2N}B_h',$$ where $B'_h$ is the block matrix,
$$	\left[\begin{array}{c}
			 B_h  \\
			\hline
			qI_N   \\
		\end{array}\right].
			$$
If we consider the previous lattice, but taking ${\mod{q}}$ (for some positive $q$), we get a (NTRU type) lattice
$${\mathcal{L}}_h^q=\{(f(x),g(x))\in {\mathcal{R}}^2 : g(x)=f(x)*h(x)\pmod{q}\}, $$ where we also write it as,
$${\mathcal{L}}_h^q=\{({\textbf f},{\textbf g})\in {\mathbb{Z}}^{2N} : [{\textbf g}]=[{\textbf f}]M_h\}$$
where 
$$	M_h=\left[\begin{array}{c|c}
			I_N & B_h  \\
			\hline
			{\textbf 0}_N & qI_N   \\
		\end{array}\right].
			$$
The lattice ${\mathcal{L}}_h$ has several interesting properties when $B_h$ is a cyclic matrix (see (\ref{B_h})). One example is when we choose $D(x)=x^N-1$. In this case, if $({\textbf a},{\textbf b})$ is a vector in the lattice, then performing a cyclic permutation of ${\textbf a}$ and ${\textbf b}$ $k-$times will result in another vector in the lattice. On the other hand, if $D(x)=x^p-x-1$ (the case of NTRU-Prime), then $B_h$ is not circulant.

For this Section, we shall proceed with the assumption that $D(x)=x^N-1$ and $p=3.$ We note that the other variants of NTRU follow a similar pattern for Encryption \& Decryption.
Alice selects public parameters $(N, p, q, d)$, with $N$ and $p=3$ being prime numbers, and both co-prime to $q.$ Usually $N$ and $q$ are large, and $q$ is a power of $2.$	
	The choice of a prime degree parameter $N$ is essential to thwart Gentry's attack, which is effective when $N$ is composite, see \cite{gentry}. 
	
	Moreover, we define the set of ternary polynomials ${\mathcal{T}}_{\alpha}$  of degree $\alpha$, as the set of polynomials with coefficients from the set $\{-1,0,1\}$ and degree at most $\alpha.$
	With ${\mathcal{T}}(d_1,d_2)\subset {\mathcal{R}},$ we denote the polynomials of ${\mathcal{R}}$ with $d_1$ entries equal to one, $d_2$ entries equal to minus one and the remaining entries are zero.	
 %When we write ${\mathcal{L}}_F$ for some polynomial $F\in {\mathcal{R}},$ we mean a subset of ${\mathcal{T}}_{\alpha}$ for some $\alpha$ relatively small with respect to $q.$

For her private key, Alice randomly selects $({\bf f}(x),{\bf g}(x))$ such that ${\bf f}(x) \in {\mathcal{L}}_f$ and ${\bf g}(x) \in {\mathcal{L}}_g$. It is important that ${\bf f}(x)$ is invertible in both ${\mathcal{R}}/q$ and ${\mathcal{R}}/3$. The inverses in ${\mathcal{R}}/3$ and ${\mathcal{R}}/q$ can be efficiently computed using the Euclidean algorithm and Hensel's Lemma, see \cite[Proposition 6.45]{hoffstein}. Let ${\bf F}_q(x)$ and ${\bf F}_3(x)$ represent the inverses of ${\bf f}(x)$ in ${\mathcal{R}}/q$ and ${\mathcal{R}}/3$, respectively.	
	
 Alice next computes
	\[{\bf h}(x) = {\bf F}_q(x) \star {\bf g}(x)\mod{q}.\]	The polynomial ${\bf h}(x)$ is Alice's public key.

The problem of distinguishing ${\bf h}(x)$ from uniform elements in ${\mathcal{R}}/q$ is called {\it decision NTRU problem.} While, the problem of finding the private key $(f(x), g(x))$ is referred to as the {\it search NTRU problem}.
	Bob's plaintext is a polynomial ${\bf m}(x) \in {\mathcal{R}},$ whose coefficients are in the set $\{-1,0,1\}.$ Thus, the plaintext ${\bf m}(x)$ is the centerlift of a polynomial in ${\mathcal{R}}/3$. Bob chooses a random ephemeral key	${\bf r}(x) \in {\mathcal{L}}_r$ and computes the ciphertext,
	\begin{equation}\label{definition_of_encryption}
	{\bf c}(x) \equiv p{\bf r}(x) \star {\bf h}(x) + {\bf m}(x) \, \bmod\,  q.
	\end{equation}
	Finally, Bob sends to Alice the ciphertext
	${\bf c}(x)\in {\mathcal{R}}/q.$
	
	To decrypt, Alice computes $${\bf v}(x) \equiv {\bf f}(x) \star {\bf c}(x) \, \bmod\,  q.$$
Then, she centerlifts ${\bf v}(x)$ to an element of ${\mathcal{R}},$ say ${\bf v}'(x),$  and she finally computes,
	\[{\bf b}(x) \equiv {\bf F}_3(x) \star {\bf v}'(x) \, \bmod\,  3.\] 
Therefore, ${\bf b}(x)$ is equal to the plaintext ${\bf m}(x)$ (this is true when a simple inequality between $d, q,$ and $N$ is satisfied).

We now briefly describe two well-known attacks based on lattices.
\subsubsection{SVP and NTRU}\label{SVPNTRU}
	Let the lattice $L_{\bf h}$ be generated by the rows of the matrix
	$$M_{\bf h}=
		\left[\begin{array}{c|c}
			I_N & C({\bf h})  \\
			\hline
			{\bf 0}_N & qI_N   \\
		\end{array}\right],
			$$ 
where $C({\bf h})$ is the matrix generated by the vector ${\bf h}$ (the public key). For the definition see (\ref{B_h}).  
			This matrix is public, since it contains the public key of the NTRU cryptosystem. From ${\bf f}(x)\star {\bf h}(x)\equiv {\bf g}(x)\pmod{q},$ there is a polynomial ${\bf b}(x)\in R$ such that $ {\bf f}(x)\star {\bf h}(x)- q{\bf b}(x)= {\bf g}(x),$ so considering polynomials as vectors we get ${\bf f}C({\bf h})-q{\bf b}={\bf g},$ thus $({\bf f},{-\bf b})M_{\bf h} = ({\bf f},{\bf g}).$ That is, $({\bf f},{\bf g})\in L_{\bf h}.$
	We see that  
	\begin{equation}\label{definition_of_lattice}
	L_{\bf h} = \{ ({\bf u}(x),{\bf v}(x))\in R^2 : {\bf u}(x)\star {\bf h}(x) \equiv {\bf v}(x) \pmod{q}\}
	\end{equation}
	or 
	\[
	L_{\bf h} = \{ ({\bf u},{\bf v})\in {\mathbb{Z}}^{2N} : {\bf u}C({\bf h}) = {\bf v} \}.
	\]
	Thus, the problem of finding the private key $({\bf f},{\bf g})\in {\mathbb{Z}}_q^{2N},$ comes down to find a short vector in the integer lattice $L_{\bf h}.$ Note that, $L_{\bf h}$ is a $q-$ary lattice, i.e. $q{\mathbb{Z}}^{2N}\subset L_{\bf h}.$ 

\subsubsection{CVP and NTRU}\label{CVPNTRU}

In this Section, we shall define the Closest Vector Problem (CVP)~\cite{gal} and its approximate variation, which is the problem we will base our attack on the NTRU system. We will also provide Babai's algorithm that solves the approximation CVP problem.

\begin{definition}
\textbf{Closest Vector Problem {\rm{(}}$CVP{\rm{(}}\L,t{\rm{)}}{\rm{)}}$}.  Given a lattice $\mathcal{L}\subset {\mathbb{Z}}^m$ of rank $n$ and a vector $\textbf{t} \in {\mathbb{R}}^m$, $\textbf{t} \notin \L$, find a vector $\textbf{u} \in \mathcal{L}$ such that, for every ${\textbf{v}} \in \mathcal{L}$ we have:
$$\| \textbf{u} - \textbf{t} \| \leq \| \textbf{v} - \textbf{t} \|$$
\end{definition}
The closest vector problem is known to be NP-hard (\cite[Chapter 3]{COMPLEXITY_OF_LATTICE_PROBLEMS}). We will now define the approximate Closest Vector Problem $(CVP_{\gamma_{n}}({\mathcal{L}},{\textbf t}))$ as follows:

\begin{definition}
\textbf{approximate-Closest Vector Problem {\rm{(}}$CVP_{\gamma_{n}}({\mathcal{L}},{\textbf t}{\rm{)}}${\rm{)}} }. \\ Given a lattice $\mathcal{L}\subset {\mathbb{Z}}^m$ of rank $n$ and a vector ${\textbf t}\in {\mathbb{R}}^m,$ find a vector ${\textbf u}\in \mathcal{L}$ such that, for every ${\textbf v}\in \mathcal{L}$ we have:
$$ \| \textbf{u} - \textbf{t} \| \leq \gamma_n \| \textbf{v} - \textbf{t} \|\  ({\rm{for \ some}} \ \gamma_n > 1).$$
\end{definition}
We say that we have a CVP oracle if we have an efficient probabilistic algorithm that solves ${\rm CVP}_{\gamma_n}$ for $\gamma_n=1.$ To solve 
${\rm{CVP}}_{\gamma_n}$, we usually use Babai's algorithm \cite[Chapter 18]{gal}, \cite{mar1} (which has a polynomial running time). In fact, combining this algorithm with LLL algorithm, we solve ${\rm{CVP}}_{\gamma_n}(\mathcal{L},{\textbf t})$ for some lattice $\mathcal{L}\subset {\mathbb{Z}}^m,$ for $\gamma_n = 2^{n/2}$ and $n=rank(\mathcal{L})$ in polynomial time. The pseudocode for Babai's algorithm is given in algorithm \ref{alg:babai}.
 
\RestyleAlgo{ruled}
\SetKwComment{Comment}{/* }{ */}
\begin{algorithm}[hbt!]
\caption{Babai's Nearest plane Algorithm}\label{alg:babai}
\SetKwInOut{Input}{input}
\SetKwInOut{Output}{output}
\Input{A $n\times m$-matrix $M$ with rows the vectors of a basis $\mathcal{B}=\{{\textbf b}_i\}_{1\leq i\leq n}$ $\subset\mathbb{Z}^{m}$ 
of the lattice $\mathcal{L}$ and a vector ${\textbf t}\in {\rm{span}}({\mathcal{B}})$}
\Output{{${\textbf x}\in L$ such that $\|{\textbf x}-{\textbf t}\|\leq 2^{n/2}dist({\mathcal{L}},{\textbf t})$.}}\ \\

$M^{*} = \{({\textbf b}_j^*)_j \}\leftarrow GSO(M)$\Comment*[r]{\small{GSO: Gram-Schimdt Orthogonalization}} 
${\textbf b} \gets {\textbf t}$\;
\For{$j=n$ {\rm{to}} $1$}{
    $c_j \gets \Big{\lfloor} \frac{{\textbf b}\cdot 
	{\textbf b}^*_j}{\|{\textbf b}^*_j\|^2} \Big{\rceil}$ \Comment*[r]{$\lfloor x 
	\rceil = \lfloor x+0.5 \rfloor$}
    ${\textbf{b}} \gets {\textbf{b}} - c_{j}{\textbf b}_j$
  }
  \Return {\bf{t}}-{\bf{b}}
\end{algorithm}
There is a relation of CVP with NTRU. The Closest Vector Problem (CVP) is not used for private key recovery but for message retrieval. To illustrate how CVP aids in recovering the message ${\bf m}$, we begin by observing that $(3{\bf r}, {\bf c} - {\bf m})$ belongs to the lattice $L_{\bf h}$ (we set $p=3$). 
This connection stems from the lattice's definition, as given in (\ref{definition_of_lattice}), and the encryption of ${\bf m}$ as outlined in (\ref{definition_of_encryption}),
$$3{\bf r}\star {\bf h} = {\bf c}-{\bf m}\pmod{q},$$ 
which implies $(3{\bf r}, {\bf c} - {\bf m})\in L_{\bf h}$.
Now, let's express it as,
$$({\bf 0}_N,{\bf c}) = (3{\bf r} - 3{\bf r}, {\bf c}-{\bf m}+{\bf m})=(3{\bf r},{\bf c}-{\bf m}) + (-3{\bf r},{\bf m}). $$
In this equation, the vector $(3{\bf r}, {\bf c}-{\bf m})$ belongs to $L_{\bf h}$, and the vector $(-3{\bf r}, {\bf m})$ is relatively short. 
We set ${\bf t}=({\bf 0}_N,{\bf c})\in {\rm{span}}({\mathcal{L}}_{\bf h}), {\bf x}=(3{\bf r},{\bf c}-{\bf m})\in {\mathcal{L}}_{\bf h}.$ Thus,
$$\|{\bf x}-{\bf t}\|=\|(3{\bf r},-{\bf m})\|\le \sqrt{10N}.$$
%To quantify this shortness (where $r(x)$ and $m(x)$ are ternary polynomials of degree at most $N-1$), we find:
%$${\rm{dist}}\big( (p {\bf r},{\bf c} - {\bf m}),({\bf 0}_N,{\bf c})\big)=||(3{\bf r},{\bf c} - {\bf m})- ({\bf 0}_N,{\bf c})||=||(3{\bf r},-{\bf m})||\le \sqrt{10N}.$$ 
Therefore, in order to find ${\bf m},$ we apply CVP in $L_{\bf h}$ with the target vector $({\bf 0}_N,{\bf c}).$ Say ${\bf z}$ is the output. Then, if the attack succeeds, the last $N$ entries of ${\bf z}$ equal to ${\bf{c}}-{\bf{m}}.$
The previous analysis can also be developed for the other two variants of NTRU. 

In the previous analysis, we are not sure, even if we have a CVP-oracle, that $(p{\bf r},{\bf c}-{\bf m})$ is the closest point to ${\bf t}.$ In fact, it can be proven that ${\rm{dist}}({\bf x},{\bf t})\approx 4.5GH({\mathcal{L}_{\bf h}})$\footnote{GH: Gaussian Heuristic, where for the specific lattice $GH({\mathcal{L}_{\bf h}})=0.35\sqrt{q}$, see \cite[section 3]{adam_draz}} (we considered $q\approx 4N$ which is true for HPS variant). In practice, this attack works only for very small $N.$ Maybe here a CVP enumeration would be more useful. One other way to attack NTRU with CVP is to consider, instead of the previous ${\bf t}$, a better target vector which approximates $(3{\bf r},{\bf c}-{\bf m})$, and so then the CVP can have better chances to find ${\bf m}.$ We shall use the previous ideas in our attack.
\subsection{Previous work-Our Contribution}
\subsubsection{Previous work}
The NTRU cryptosystem was first targeted for a lattice attack in 1997 by Coppersmith \cite{CopSha97}. Later, Gentry developed an effective strategy, particularly useful when $N$ is composite, see \cite{gentry}. May, in \cite{May}, used his own class of lattices, called run-lattices, to solve similar issues. 

Silverman \cite{Silverman} built on May's idea and proposed a technique that involves choosing $r$ coefficients and reducing the dimension of the lattice to make them zero. A similar approach was taken by the authors of \cite{Gama-Nguyen}, who leveraged decryption failures to recover the secret key if the decryption oracle supported it.

Other techniques have been developed, such as reducing the NTRU problem to a multivariate quadratic system over a finite field with two elements using Witt vectors \cite{gerald,witt}. Odlyzko \cite{Odlyzko} proposed a meet-in-the-middle attack that divides the search space into two parts, resulting in lower time complexity but high memory requirements. Howgrave's hybrid attack \cite{hybrid attack} combines lattice reduction with a meet-in-the-middle algorithm, and has been widely studied by researchers to evaluate the security of lattice-based encryption methods.

To attack NTRU cryptosystem with a higher modulus than that in the NTRUEncrypt standard, similar methods were proposed independently by Albrecht, Bai, and Ducas \cite{Albrecht} and also by Cheon, Jeong, and Lee \cite{Cheon}. Kirchner \cite{Kirchner} demonstrated that time complexity is polynomial for $q=2^{\mathrm {\Omega }(\sqrt{n\log \log n})}$ in $\mathbb {Q}(\zeta _{2^n}).$ Lastly, Nguyen \cite{Nguyen - Boosting the hybrid} enhanced and clarified the hybrid and meet-in-the-middle attacks. The subfield attack variation proposed in this paper is superior to the earlier methods, but not as effective as the hybrid attack.

In 2005, Silverman, Smart, and Vercaturen introduced another line of research using Witt's vectors, see \cite{witt}, later expanded by Bourgeois and Faugère, see \cite{gerald}. 

In their recent work, documented in the paper \cite{may_recent}, the authors present a new approach to tackling the latest iterations of the NTRU encryption scheme. Their method involves utilizing a carefully constructed lattice and employing the BKZ algorithm in conjunction with the lattice sieving algorithm from the G6K library. A key highlight of their research lies in the compelling advantages derived from transitioning away from the traditional Coppersmith-Shamir lattice towards a basis rooted in the cyclotomic ring. This strategic shift yields remarkable results, as evidenced by their ability to decrypt the HPS-171 instance in a mere $83$ core days, utilizing the cyclotomic ring basis, as opposed to the $172$ core days required when using the Coppersmith-Shamir basis.  Additionally, the authors take on another official NTRU challenge, specifically one with $N = 181,$ which had been established by Security Innovation, Inc. Their approach allows them to successfully crack this challenge in $20$ core years.
\subsubsection{Our contribution}
In \cite{adam_draz} the authors provide a message recovery attack based on CVP. 
The idea was to shift the CVP attack from its original setting of $[{\mathcal{L}}_{\bf h},{\bf t}=({\bf 0}_N,{\bf c})]$ to a different CVP instance  $[{\mathcal{L}}_{{\bf{\alpha}}},{\bf t}'],$ where ${\bf{\alpha}}$ represents a polynomial ${\alpha}(x)$ that we have the freedom to select. By suitably choosing certain polynomials ${\alpha}(x),$ the authors successfully executed message recovery attacks on NTRU-HPS.  While an oracle was used in the attack (which poses a drawback for the method), employing side-channel attacks might aid in initiating such an oracle. The authors also used an approximation version of CVP through Babai's algorithm.

In the present paper, we first replace the approximate CVP-algorithm of \cite{adam_draz} with an exact CVP-algorithm, which allows for more precise and better results of our attack. This algorithm will run in polynomial time. The main idea is to reduce our attack to a lattice, which is Voronoi First Kind (VFK), in which there is an exact CVP polynomial algorithm, \cite{VFK-CVP}. Additionally, we apply our attack to a different lattice than in \cite{adam_draz}, in fact we consider the case where $\alpha(x)$ is a constant.

Secondly, our attack is equally applicable to all three variants (HPS, HRSS, Prime) without any dependence on the underlying ring structure. In fact, the target vector effectively encapsulates all the essential information for the specific NTRU instance.

Thirdly, we executed our attack on both NTRU-HPS and NTRU-Prime, detailing our findings in Section \ref{sec:experiments}. Our results demonstrate substantial improvements over the prior experiments outlined in \cite{adam_draz}. Although we didn't implement the attack for NTRU-HRSS, our results are akin to those for NTRU-HPS.

%Additionally, the use of VFK lattices, we observe an improvement in the simplest Babai attack. This is particularly evident in the variant \texttt{ntruhps2048509}, where our method is shown to be superior to Babai's attack in the same lattice. These advancements significantly  enhances the previous attack \cite{adam_draz}. 
Furthermore, our research has uncovered a non-trivial connection between NTRU lattice and VFK lattices, which enhances our understanding of NTRU encryption system. Despite relying on an oracle, which we'll elaborate on shortly, the proposed attack offers a way to evaluate the resilience of NTRU-based cryptosystems against message recovery attacks. Our attack includes a parameter that, when large, indicates a system's weakened resistance. By conducting experiments on NTRU cryptosystems and estimating this parameter, say $R_0$, we can effectively grade the system's security level against these types of attacks.

Finally, the lattice we used in our attack is independent of NTRU parameters, such as the public key $h(x)$. This characteristic allows for a singular pre-calculation. For instance, computing a superbasis is a one-time task, regardless of the message changes. On top of that, we also delved into the CVP variant tailored for VFK lattices. Our study led us to develop both pseudocode and a Sagemath implementation \cite{sage} (see Appendix \ref{appendixA}), which we present herein.
\subsection{Roadmap}

Section \ref{sec:VFK} provides a brief overview of integer lattices and further offers essential background information on Voronoi First Kind lattices, as well as an overview of the polynomial algorithm that solves the CVP problem on VFK lattices. In Section \ref{ntru-variants}, we delve into the details of the NTRU encryption system and its diverse variants. Section \ref{sec:Auxiliary} presents critical auxiliary results that underpin our subsequent discussion. In Section \ref{sec:attack} we outline our proposed attack strategy. Then, in Section \ref{sec:experiments}, we detail our experimental methodology and findings. Finally, the concluding Section summarizes key insights and contributions.

Our work's corresponding implementation can be found at \url{https://github.com/drazioti/ntru_msg_recovey_attack}.
\section{Voronoi First Kind Lattices}\label{sec:VFK}

In this Section, we recall some well-known facts about lattices which form the 
background to our algorithms.

Let ${{\textbf{b}}_1,{\textbf{b}}_2,\ldots,{\textbf{b}}_n}$ be linearly independent vectors of 
${\mathbb{R}}^{m}$.
The set 
\[\mathcal{L} = \bigg{\{} \sum_{j=1}^{n}\alpha_j{\textbf{b}}_j :
 \alpha_j\in\mathbb{Z}, 1\leq j\leq n\bigg{\}}\]
is called  a {\em lattice} and 
the set $\mathcal{B} = \{{\textbf{b}}_1,\ldots,{\textbf{b}}_n\}$ a basis of 
$\mathcal{L}$. 
All the bases of $\mathcal{L}$ have the same number of elements, $n,$ which is  called
 {\em dimension} or {\em rank} of $\mathcal{L}$. If $n=m$, then
the lattice $\mathcal{L}$ is said to have {\em full rank}. 
We denote by $M$ the $n\times m$-matrix  having as rows the vectors 
${\textbf{b}}_1,\ldots,{\textbf{b}}_n$. 
If $\mathcal{L}$ has full rank, then the {\em volume} of the lattice
 $\mathcal{L}$ is defined to be the positive number
 $|\det{M}|$ which is independent from the basis $\mathcal{B}$.
 If ${\textbf v}\in \mathbb{R}^m$, then $\|{\textbf v}\|$ denotes, as usually, the
 Euclidean norm of ${\textbf v}$. 
 We denote by $LLL(M)$, the application of well-known LLL-algorithm on the rows of $M$.
 Finally,   we denote by 
$\lambda_1(\mathcal{L})$ the smallest of the lengths of vectors in 
$ \mathcal{L}-\{ \textbf 0\}$.

First, we shall define the Voronoi Cell and the relevant vectors of a Voronoi First Kind (VFK) Lattice, and afterwards we will express the criteria under which a lattice can be considered as VFK. Finally, we will describe an efficient algorithm~\cite{VFK-CVP} to solve the CVP problem for VFK Lattices.

\begin{definition}
% cite https://arxiv.org/pdf/1512.00720.pdf, page 2
\textbf{Voronoi cell Vor($\L$)}: Given a lattice $\L \subset \mathbb{R}^m$, the Voronoi cell $Vor(\L)$, is the subset of $\mathbb{R}^m$ containing all points closer or of equal distance to the lattice point at the origin than to any other lattice point (with respect to the given norm):
$$ Vor(\L) = \big\{ \textbf{x} \in \mathbb{R}^m: \|\textbf{x}\| \leq \| \textbf{x}-\textbf{v} \|, \ \forall \textbf{v} \in \L \backslash \{\bf{0}\} \big\},$$
\end{definition}

A lattice point $\textbf{v}$ is called relevant if and only if $\textbf{0}$ and $\textbf{v}$ are the only closest vectors to $\frac{1}{2} \textbf{v}$, see  \cite{voronoi}.
%Vectors $\textbf{v}$ for which there exists a point that is closer to $ \textbf{v}$ and the origin ${\bf 0}$ than to any other lattice vector, are called \textbf{relevant vectors} and are denoted by $Rel(\textbf{v})$. We are now ready to define the VFK Lattice.

\begin{definition}
(\textbf{Voronoi First Kind Lattice}). A lattice $\L$ of dimension $n$ is called Voronoi First Kind (VFK) if and only if there is a lattice basis $\{{\textbf b}_i\}_{i=1,\dots,n}$ and a lattice vector ${\textbf b}_{0}$ such that:\\\\
$
\ ({\rm{\bf{i}}}) \sum_{i=0}^{n}{\textbf b}_i ={\textbf 0}\ (\text{superbasis\ condition)}
$ \ \\
$
({\rm{\bf{ii}}}) \ q_{ij} = {\textbf b}_i\cdot {\textbf b}_j \leq 0, \ \text{for} \ i,j=0,1,\dots,n \ \text{and} \ i\not =j\ (\text{obtuse condition}).
$

\end{definition}

The $n+1$ vectors $\{{\textbf b}_i\}_{i=0,...,n}$ make the obtuse superbasis of the lattice $\L$. The $\{q_{ij}\}$ are called the Selling parameters~\cite{SELLING}. An important and useful property of VFK lattices is the way their relevant vectors can be written as the sum of a set of their superbasis vectors~\cite{CONWAY}:
\begin{equation}
\sum_{i\in I}{\textbf b}_i,
\label{equation:VFK_REL}
\end{equation}
where $I$ is a strict subset of $ \{ 0,1,..,n \} $ and $I \neq\emptyset$. We can now start explaining the main algorithm that we used for our attack.

\subsection{CVP efficient algorithm for VFK lattices}

Let $\L$ be a $n$-dimension VFK lattice in $\mathbb{R}^m$, with an obtuse superbasis $\{{\textbf b}_i\}_{i=0,...,n}$ and \textbf{y} a point in ${\rm{span}}(\L)\subset \mathbb{R}^m.$ We also define $B$ be a $ (n+1) \times m $ matrix whose rows are represented by the obtuse superbasis vectors. Since ${\bf y}\in span({\mathcal{L}}),$ there exists $\textbf{z}\in {\mathbb{R}}^{n+1}$, such that $\textbf{y}=\textbf{z}B$. In order to solve the CVP problem in VFK lattices, we need to find a vector ${\bf w}=(w_0,w_1,...,w_n)$ that minimizes the norm:

\begin{equation}
\Big{\|} \textbf{y} - \sum_{i=0}^{n}{w_i \textbf{b}_i } \Big{\|}^2 = \Big{\|} (\textbf{z}-{\textbf{w}})B \Big{\|}^2\label{equation:VFK_CVP}
\end{equation}

McKilliam et al.~\cite{VFK-CVP} proposed an interesting iterative polynomial algorithm that solves the CVP. Assuming that $\textbf{x}_0$ is a lattice point in $\L$, the iterative procedure is the following:
\begin{align*}
\textbf{x}_{k+1} &= \textbf{x}_k + \textbf{v}_k, \\
\textbf{v}_k &= \underset {\textbf{v} \in {\rm{Rel}}(\L) \cup \{{\textbf 0}\}}{\rm{argmin}} \| \textbf{y}-\textbf{x}_k-\textbf{v} \|,
\end{align*}
where $\textbf{v}_k$ is the relevant vector $\textbf{v}$ that minimizes the norm $ \| \textbf{y}-\textbf{x}_k-\textbf{v} \|$, during the $k$-th iteration. There might be many relevant vectors that minimize this norm, but the choice does not change the outcome of the algorithm. It is proven by the same authors that after every iteration of the algorithm, a closer point is found; that is 
\begin{equation}
\|\textbf{y}-\textbf{x}_k \| > \| \textbf{y}-\textbf{x}_{k+1}\|. 
\end{equation}
After at most $n$ iterations of the algorithm, the closest vector will be found~(\cite[Propositions 3.1, 3.2]{VFK-CVP}). Choosing a $\textbf{x}_0$ that is a close approximation of the closest point can influence greatly the number of iterations the algorithm will perform but finding such a point can be computationally difficult~\cite{hard-approx-cvp}. Nevertheless, in \cite[chapter 5]{VFK-CVP}, it is proven that this part of the algorithm has a worst-case complexity of $\mathcal{O}(n)$ for VFK lattices. If we take into consideration the matrix $B$, the vector ${\bf y}$, and the property \eqref{equation:VFK_CVP}, then the iterative procedure becomes:
\begin{align*}
\textbf{x}_{k+1} &= \textbf{u}_{k+1}B, \\
\textbf{u}_{k+1} &= \textbf{u}_k + \textbf{t}_k, \\
\textbf{t}_k &= \underset {\textbf{t} \in \{0,1\}^{n+1}}{\rm{argmin}} \| (\textbf{z}-\textbf{u}_k-\textbf{t})B \|, 
\end{align*}
Finally, according to property \eqref{equation:VFK_REL}, a relevant vector $\textbf{v}$ can be written as:
\\
$$\textbf{v} = \sum_{i = 0}^{n} t_i\textbf{b}_i,$$ 
\\
where $t_i$ is either $0$ or $1$.

The maximum number of relevant vectors of a lattice is exponential $2^{n+1}-2$, thus excluding a brute force approach. Therefore, the minimization problem was handled in a different way.

% $v_k$ is found by computing a minimum cut in an undirected flow network. The graph is created because of the VFK lattice property \eqref{equation:VFK_REL}.

We set ${\bf p}=\textbf{z}-\textbf{u}_k=(p_i)$ and so the norm in \eqref{equation:VFK_CVP} becomes:
\begin{equation}
\begin{split}
\Big{\|} ({\bf p}-\textbf{t})B \Big{\|}^2 & =  \Big{\|} \sum_{i=0}^{n}{\textbf{b}_i(p_i-t_i)} \Big{\|}^2 \\ & = \sum_{i=0}^{n}\sum_{j=0}^{n}q_{ij}p_ip_j - 2\sum_{i=0}^{n}\sum_{j=0}^{n}q_{ij}p_j t_i + \sum_{i=0}^{n}\sum_{j=0}^{n}q_{ij}t_it_j,
\end{split}
\label{equation:VFK_CVP_MIN}
\end{equation}
where $q_{ij}={\bf b}_i\cdot {\bf b}_j$.

The first term of \ref{equation:VFK_CVP_MIN} is constant and does not affect the minimization problem, meaning that it can be ignored. In addition, by setting  $s_i = -2 \sum_{j=0}^{n}{q_{ij}p_j}$, the above quadratic form becomes:

%Let $s_i = - 2\sum_{i=0}^{n}\sum_{j=0}^{n}q_{ij}p_j$ and the above quadratic form becomes:
\begin{equation}
Q(t) = \sum_{i=0}^{n}s_i t_i + \sum_{i=0}^{n}\sum_{j=0}^{n}q_{ij}t_it_j.
\label{equation:VFK_QT}
\end{equation}

% add citation next to laplacian matrix [5,13]
% see and add citation 44,49,53,12 about the technique with the flow networks and mincuts
In \cite{VFK-CVP}  the authors have shown that a minimizer of $Q(t)$ can be found by computing a minimum cut in an undirected flow network, since there is a one-to-one correspondence between the obtuse superbasis of the VFK lattice and the Laplacian matrix. We recall the following definitions related to the maximum flow problem.

\begin{definition}\label{def:flowNetwork}
An undirected flow network is a graph $G$ with $n+3$ vertices $v_i$ and edges $e_{ij}$ connecting $v_i$ to $v_j$ with weight $w_{ij}$. The graph is undirected, therefore $w_{ij}=w_{ji}$. The first vertex, $v_0,$ is the source vertex while the last one, $v_{n+2}$, is the sink vertex.
\end{definition}

\begin{definition}\label{def:cut}
A cut is a subset $C \subset V$, where $V$ is the set of vertices in a graph, and its complement $\Bar{C} \subset V$, where $v_0 \in C, v_{n+2} \in \Bar{C}$.
\end{definition}

\begin{definition}\label{def:cutWeight}
The weight $W$ of a cut is the sum of the weights on the edges crossing from the vertices in $C$ to the vertices in $\Bar{C}$:
$$W = \sum_{i \in I} \sum_{j \in J} w_{ij},$$
where $I=\{ i | v_i \in C\}, J=\{ j | v_j \in \Bar{C}\}$.
\end{definition}

\begin{definition}\label{def:mincut}
Mincut is a pair $\{C,\Bar{C}\}$ that minimizes the weight $W$.
\end{definition}

Let's define an undirected flow network $G$, as shown in definition \ref{def:flowNetwork} and its minimum cut $C$ (definitions \ref{def:cut} and \ref{def:mincut})/ that has a weight $W$, see definition \ref{def:cutWeight}. It was shown in (\cite{VFK-CVP}, chapter 6) that the weight $W$ can be represented with the following quadratic form:

\begin{equation}
\begin{split}
F(t) & = \sum_{i=0}^{n}g_it_i - \sum_{i=0}^{n}\sum_{j=0}^{n}w_{ij}t_it_j, \\
g_i & = d_i + \sum_{j=0}^{n}w_{ij}, \\
d_i & = w_{i,n+2} - w_{0,i}.
\end{split}
\label{equation:VFK_FT}
\end{equation}
From the equivalence of \eqref{equation:VFK_QT} and \eqref{equation:VFK_FT} we can see that the weights of our graph $G$ can be calculated using the Selling parameters $q_{ij}$ and $s_i$:

\begin{equation}
\begin{split}
w_{i,j} & = -q_{ij} \ {\rm{(Selling \ parameters)}}, \\
s_i & = d_i + \sum_{j=0}^{n} w_{ij}
\end{split}
\label{equation:VFK_S}
\end{equation}
Finally, the weights can be chosen so that they are non-negative: 
\begin{equation}
\begin{split}
w_{0,i} & = 0 \ {\rm{if}} \ s_i < 0, \ {\rm{else}} \ w_{0,i} = - s_i \\
w_{i,n+2} & = s_i \ {\rm{if}} \ s_i \geq 0, \ {\rm{else}} \  w_{i,n+2} = 0
\end{split}
\label{equation:VFK_WEIGHTS}
\end{equation}
We create a vector $\textbf{t}=(t_0,...,t_n)$ from set $C$ such that:
\begin{equation}
t_i = 1, \ {\rm{if}} \ i \in C, \ {\rm{else}} \  t_i = 0
\label{equation:VFK_WEIGHTS_2}
\end{equation}

From the above discussion, we deduce that if we do a minimum cut in the graph $G$ with $n+3$ vertices and choose the weights of the edges according to \eqref{equation:VFK_S} and \eqref{equation:VFK_WEIGHTS}, then the vertices in $C$ determine the vector ${\bf t}$ of the iterative procedure. This minimization process requires $\mathcal{O}(n^3)$ operations, see ~\cite{intro-algos, mincut-problems}. Since the number of iterations is $\mathcal{O}(n)$, the total time complexity of the algorithm is $\mathcal{O}(n)*\mathcal{O}(n^3)=\mathcal{O}(n^4)$.

The pseudocode of the algorithm that we used for our experiments can be found in Appendix~\ref{appendixA}. We continue with the presentation of all NTRU variants.

\section{NTRU-KEM}\label{ntru-variants}
We first define a Key Encapsulation Mechanism (KEM), which consists from three algorithms.
\begin{itemize}
\item[{\bf{keyGen}}] A probabilistic algorithm that outputs a public key({\it{pk}}) and a secret key({\it{sk}}).
\item[{\bf{Encap}}] A probabilistic algorithm that takes as input the public key and outputs a ciphertext($c$) and a shared secret($k$)\footnote{The ciphertext is called {\it{encapsulation of the key} $k$.}}. 
\item[{\bf{Decap}}] A probabilistic algorithm that takes as input the ciphertext and the secret key, and outputs the shared secret. 
\end{itemize}
Let $q,N$ be positive integers with $N$ prime and $D(x)$ be a polynomial of degree $N.$ 
We set \\\\
$-$ ${\mathcal{R}}={\mathbb{Z}}[x]/\langle D(x) 
\rangle$, ${\mathcal{R}/3}={\mathbb{Z}}_3[x]/
\langle D(x) \rangle$ and  
${\mathcal{R}}/q={\mathbb{Z}}_q[x]/\langle D(x) 
\rangle$. \\\\
$-$ ${\mathcal{S}}={\mathbb{Z}}[x]/\langle \Phi_N(x) 
\rangle$, ${\mathcal{S}/3}={\mathbb{Z}}_3[x]/\langle \Phi_N(x)  \rangle$ 
and  ${\mathcal{S}}/q={\mathbb{Z}}_q[x]/\langle \Phi_N(x) \rangle.$\\\\
A ternary polynomial is one that has as coefficients only the integers $-1,0,1.$ With 
${\mathcal{T}}_a$ we denote the set of ternary polynomials of ${\mathcal{R}}$ with degree at most $a$ and $\mathcal{T}_a(d_1,d_2)\subset {\mathcal{T}}_{a}$ consists from elements of ${\mathcal{T}}_{a}$ with $d_1$ coefficients equal to $1$ and $d_2$ equal to $-1.$ Furthermore, with ${\mathcal{T}}_a(w), w\in{\mathbb{Z}}_{>0}$ we denote the ternary polynomials which have $w$ non-zero coefficients. With ${\mathcal{T}}_{a,+}$ we denote the subset of ${\mathcal{T}}_a$ consisting from polynomials $f(x)=\sum_{i=0}^{\alpha}v_ix^i$ such that $\sum_{i=0}^{a-1}v_iv_{i+1}\ge 0.$	We also have four sample spaces,
${\mathcal{L}}_f, {\mathcal{L}}_g, {\mathcal{L}}_r,$ and ${\mathcal{L}}_m.$
We set, $\Phi_1(x)=x-1$ and in the case of $D(x)=x^N-1,$ we set $\Phi_N(x) = D(x)/\Phi_1(x).$  
\subsection{Generation of Public key}
We choose the {\it {secret key}}, $(g(x), f(x))$ from ${\mathcal{L}}_g$ and ${\mathcal{L}}_f$, respectively, which both are suitable subsets of ${\mathcal{T}}_{a}$ (for some $a$). Also $g(x)$ is invertible in ${\mathcal{R}}/3.$  The {\it {public key}} is $h(x)=g(x)/(\varepsilon f(x))$ invertible in ${\mathcal{R}}/q,$ where $\varepsilon=1/3$ or $3.$ Since there are many different sample spaces in the bibliography, here for NTRU-HPS and NTRU-HRSS, we follow \cite{nist2-ntru} (also for HRSS see \cite{HRSS,HRSS-2}) and for NTRU-Prime we follow \cite{ntru-prime}. In all the flavors we ignore encoding/decoding operations, except in NTRU-prime where we use a simple encoding scheme.
\subsection{Sample spaces}
\subsubsection{NTRU-HPS} $N$ is prime, and $q\le 16N/3+16$ which is a power of two, $\varepsilon=1/3,$ $D(x)=x^N-1.$\\
$-$ ${\mathcal{L}}_m={\mathcal{L}}_g=
{\mathcal{T}}_{N-2}(\frac{q}{16}-1,\frac{q}{16}-1),$\\
$-$ ${\mathcal{L}}_f={\mathcal{L}}_r={\mathcal{T}}_{N-2}.$
\subsubsection{NTRU-HRSS}\label{subsection:ntru-hrss} $N$ is prime and $q = 2^{\lceil 7/2+\log_2{N}\rceil} > 8\sqrt{2}(N-1)$, $\varepsilon=1/3,$ $D(x)=x^N-1,$\\
$-$ ${\mathcal{L}}_m={\mathcal{L}}_r={\mathcal{T}}_{N-2},$\\
$-$ ${\mathcal{L}}_f={\mathcal{T}}_{N-2,+},$\\
$-$ ${\mathcal{L}}_g=\{(x-1)*v(x) : v(x)\ {\text{in}}\ {\mathcal{T}}_{N-2,+} \}.$\\\\
Remark that NTRU-HRSS uses arbitrary weight sample spaces. For instance, in NTRU-HPS the message has weight $q/8-2,$ so the corresponding vector ${\bf m}$ has Euclidean length $\sqrt{q/8-2}.$
\subsubsection{NTRU-prime}
NTRU-prime is different from the two previous flavors, since in this case ${\mathcal{R}}$ is a field, whilst in the previous cases, it was a ring but not a field. We set $\varepsilon=3$ and let $w,p,q$ positive integers with $p,q$ primes and such that $w\ge 1, p\ge 1.5w, q\ge 16w+1.$ We set
$D(x)=x^{p}-x-1,$ which is irreducible $\mod{ q}.$ \\
$-$ ${\mathcal{L}}_g={\mathcal{T}}_{p-1},$ ${\mathcal{L}}_f={\mathcal{L}}_r={\mathcal{T}}_{p-1}(w).$ \\
Here we do not use ${\mathcal{L}}_m,$ but this set implicitly defined. We choose randomly $g(x)$ from ${\mathcal{L}}_g$ invertible in ${\mathcal{R}}/3,$ and $f(x)$ randomly from ${\mathcal{L}}_f$.
The public key is $h(x)=g(x)/(3f(x))$ in ${\mathcal{R}}/q.$ Remark that, $w>0$ and ${\mathcal{R}}/q$ is a field, thus polynomial $h(x)$ is invertible in ${\mathcal{R}}/q.$
 
\subsection{Encap/Decap for NTRU-HPS/HRSS}
 We provide Encapsulation/ Decapsulation algorithms for NTRU-HPS and NTRU-HRSS. The algorithm for these two flavors are the same and only the sample spaces are different. We define the function ${\rm{Lift}_3}:{\mathcal{L}}_m\rightarrow {\mathcal{S}}/3$ be the reduction in the ring ${\mathcal{S}}/3$ i.e the unique polynomial of degree at most $N-2$ with coefficients from the set $\{-1,0,1\}.$ For the HRSS the analogous function is ${\rm{Lift}'_3}:{\mathcal{L}}_m\rightarrow {\mathcal{S}}/3$ such that ${\rm{Lift}'_3}(a(x))=(x-1)*{\rm{Lift}_3}(a(x)/(x-1)).$ For a pseudocode that implements ${\rm{Lift}'_3}$, see \cite[Appendix 2 - Algorithm 8]{HRSS-2}.\\\\
 The {\it GenKey} is common to NTRU-Encrypt and NTRU-KEM. We first describe NTRU-Encrypt.\\\\
 {\it GenKey}. Takes as input a {\it{seed}} and outputs a quadruple $({\textbf h},({\textbf f},{\textbf f}_3,{\textbf h}_q)).$\\\\
\texttt{1.}   $(f(x),g(x))\xleftarrow{\$} {\mathcal{L}}_f\times {\mathcal{L}}_g$ \hspace{10pt}\\
\texttt{2.}  $f_q(x)\leftarrow f^{-1}(x)\mod{(q,\Phi_N(x))}$\\
\texttt{3.}   $f_3(x)\leftarrow f^{-1}(x)\mod{(3,\Phi_N(x))}$\\
\texttt{4.}   $h(x) \leftarrow 3g(x)*f_q(x)\mod{(q,D(x))}$ \ $\# D(x)=x^N-1$\\
\texttt{5.}  $ h_q(x) \leftarrow h^{-1}(x)\mod{(q,\Phi_N(x))}$\ \\  
\texttt{6.}  $ {\mathbb{S}} \xleftarrow{\$} \{0,1\}^{256}$\ \\ 
\texttt{7.}   {\textbf{return}} \ $(pk,sk)=({h(x)},({f(x)},{f_3(x)},{h_q(x)},{\mathbb{S}}))
 $\\\\
 {\it Encryption}. Takes as input the the public key $h(x)$ and the pair $(r(x),m(x)),$ where $r(x)\xleftarrow{\$}{\mathcal{L}}_r$ and returns the ciphertext $c(x)$.\\\\
 \texttt{1.} $m'(x)= {\rm Lift}_3(m(x))$ \ (or ${\rm{Lift}'_3}(m(x))$ for HRSS)\\
 \texttt{2.} $ c(x) \leftarrow h(x)*r(x) + m'(x) \mod{(q,\Phi_1\Phi_N)}$ \\
 \texttt{3.} {\textbf{return}} $c(x)$\\\\
 {\it Decryption}. Takes us input the secret key $\big(sk=(f(x),f_3(x),h_q(x)),c(x)\big)$ and returns the pair $(r(x),m(x))$ or $(0,0,1).$\\\\
\texttt{1.} {\textbf{if}} $c(x) \not\equiv 0 \mod{(q,\Phi_1)}$\\
\texttt{2.}  \hspace{10pt} {\textbf{return}}  \ $(0,0,1)$ \\
\texttt{3.} $ a(x) \leftarrow c(x)*f(x) \mod{(q,\Phi_1\Phi_N)}$\ \\
\texttt{4.} $ m(x) \leftarrow a(x)*f_3(x) \mod{(3,\Phi_N)}$\ \\
\texttt{5.} $m'(x) \leftarrow {\rm Lift}_3(m(x)) $\\
\texttt{6.} $r(x) \leftarrow (c(x)-m'(x))h_q(x) \mod{(q,\Phi_N)}$\\
\texttt{7.} {\textbf{if}} $(r(x),m(x)) \in {\mathcal{L}}_r\times {\mathcal{L}}_m$ {\textbf{then}}\\
\texttt{8.} \hspace{10pt} {\textbf{return}} $(m(x),r(x),0)$\\
\texttt{9.}  {\textbf{else}}\ \\
\texttt{10.}  \hspace{10pt} {\textbf{return}}  \ $(0,0,1)$ \ \\\\
From the previous design we get,
$${\texttt{Encrypt}}((r(x), m(x)), pk) = ct \Leftrightarrow {\texttt{Dec}}(ct(x), sk) = (r(x), m(x)).$$
We now describe the encpasulation-decapsulation functions.\\\\
{\it Encapsulation}. Takes as input a seed, the {\it{pk}} and outputs a pair, ciphertext $c(x)$ and shared secret $s.$\\\\
\texttt{1.} $(r(x),m(x)) \xleftarrow{\$} {\mathcal{L}}_r\times {\mathcal{L}}_m$ \\
\texttt{2.} $ c(x) \leftarrow {\rm{Encrypt}}({h(x)},{r(x)},{m(x)})$\\
\texttt{3.}  $s \leftarrow {\rm{hash}}({r(x)},{m(x)})$\\
\texttt{4.} {\textbf{return}} $({c(x)},s)$\\\\
 {\it Decapsulation}. It takes as input the ({\it{sk}}, ${\mathbb{S}}$) and the ciphertext, and outputs the shared secret $s.$\\\\
 \texttt{1.}  $({m(x)},{r(x)},fail) \leftarrow {\rm{Decrypt}}(sk,{c(x)})$\\
 \texttt{2.}  $ s \leftarrow {\rm hash}({r(x)},{m(x)})$\\
\texttt{3.} ${\textbf{if}} {\ \rm fail}=0\ {\textbf{then}}$ \\
\texttt{4.} \hspace{10pt} {\textbf{return}}  $s$ \\
\texttt{5.}  {\textbf{else}}\ \\
\texttt{6.}  \hspace{10pt} {\textbf{return}}  \ $hash({\mathbb{S}},c(x))$ \\
\ \\ {\it{fail}} is a boolean function that accepts $({\textbf c},{\textbf r},{\textbf m}),$ and returns $0$ if the decryption is correct.

\subsection{Encap/Decap for NTRU-Prime}\label{kem-ntru-prime} Let $p,q$ primes and $w=2t$ an even positive integer. Also, $D(x)=x^p-x-1.$ We need the following sets,
 if $q\equiv 1\pmod{3}$ we define
 $$A_{q,1}=\{-(q-1)/2,...,-6,-3,0,3,6,...,(q-1)/2\},$$ 
 and if $q\equiv 2\pmod{3},$ then
 $$A_{q,2}=\{-(q+1)/2,...,-6,-3,0,3,6,...,(q+1)/2\}.$$ 
 In both cases $A_{q,1}, A_{q,2}$ are subsets of $3{\mathbb{Z}}.$
  {\texttt{Round\_set}} is the set of all polynomials 
 $${\texttt{Round\_set}}=\bigg\{ \sum_{j=0}^{p-1} r_jx^j\in {\mathcal{R}}\ {\rm{with}}\ r_j\in A_{q,1}\ {\rm{if}}\ q\equiv 1 \pmod{3},$$ 
 $${\rm{else}}\ r_j\in A_{q,2}\ {\rm{if}}\ q\equiv 2 \pmod{3}\bigg\}.$$
 Also, we consider a function ${\texttt{closest}_3}$:${\mathbb{Z}}\rightarrow 3{\mathbb{Z}},$ where ${\texttt{closest}_3}(x)$ is the multiple of $3$ nearest to $x.$ I.e. is defined  
\[ {\texttt{closest}}_3(x) = 3\times {\rm{sgn}}(x) \big\lfloor|x/3|+0.5\big\rfloor.\] 
If we write the function ${\rm{sgn}}(x) \big\lfloor|x/3|+0.5\big\rfloor$ as ${\texttt{round}}(x)$ then
\[ {\texttt{closest}}_3(x) = 3{\texttt{round}}(x).\]
 Now, the function ${\texttt{Round}}:{\texttt{Round\_set}} \rightarrow 3{\mathbb{Z}}$ is such that ${\texttt{Round}} = {\texttt{closest}_3}\circ {\rm{Lift}}_q,$
where ${\rm{Lift}}_q:{\mathcal{R}}/q\rightarrow {\texttt{Round\_set}}$ and ${\texttt{closest}_3}$ acts in the coefficients. 
\begin{remark}\label{tricky} We remark that ${\texttt{Round}}(g(x))=g(x)+e(x)$ for some $e(x)$ ternary polynomial. To see this, consider the function $\mu:{\mathbb{Z}}\rightarrow \{\texttt{ternary polynomials}\}$ defined by $\mu(x)=x-{\texttt{closest}}_3(x).$ Extend $\mu$ to 
${\mathcal{R}}/q$ say 
$$\overline{\mu}:{\mathcal{R}}/q \rightarrow \{\texttt{ternary polynomials}\}$$ and 
$${\overline{\mu}}\bigg(\sum_{j=0}^{q-1} r_jx^j\bigg)=\sum_{j=0}^{q-1}\mu(r_j)x^j.$$ 
Then, $e(x)=-\overline{\mu}(g)(x).$
\end{remark}\ \\
{\it GenKey}. 
Takes as input a {\it{seed}} and outputs a triple $({\textbf h},({\textbf f},{\textbf g}_3)),$ where ${\textbf h}$ is the public key and $({\textbf f},{\textbf g}_3)$ the secret key.\\\\
\texttt{1.}  $f(x)\xleftarrow{\$} {\mathcal{L}}_f={\mathcal{T}}_{p-1}(w) $  \hspace{10pt} \# {\texttt{also} $T_{p-1}(w)$ \texttt{is called Short}}\\
\texttt{2.}  $g(x)\xleftarrow{\$} {\mathcal{L}}_g={\mathcal{T}}_{p-1}\hspace{10pt}\ \ \ \ \ \ \# \ {\texttt{with}}\ g(x)\  \texttt{invertible modulo}\ 3.$  \\
\texttt{3.}  $g_3(x)\leftarrow g^{-1}(x)\mod{(3,D(x))}$\\
%\texttt{4.} $f_q(x)\leftarrow f^{-1}(x)\mod{(q,D(x))}$\\
\texttt{4.}  $h(x) \leftarrow g(x)/(3f(x))\mod{(q,D(x))}$\\ 
\texttt{5.} {\textbf{return}} \ $(pk,sk)=({\textbf h},({\textbf f},{\textbf g}_3))$\ \\\\
{\it Encapsulation}.
\ \\
For encapsulation we do not follow exactly the code of NTRU-prime. For instance, we did not use the encoding/decoding algorithms but some simplification of these algorithms. In our attack we need the ciphertext, i.e., the polynomial $ct(x),$ which we have it (after decoding it) from the original NTRU-prime. If we manage to find $m(x)$ we can compute $r(x)$ and then the shared key.

Encapsulation  takes as input a seed and the public key and outputs a ciphertext $ct(x)$ and a shared key ${\textbf K}$.\\\\
\texttt{1.}  $r(x)\xleftarrow{\$} {\mathcal{L}}_r$ \hspace{10pt} \# {\texttt{here}}\ $L_r \texttt{ is } {\mathcal{T}}_{p-1}(w)$ \texttt{where} $w$ \texttt{is even}\footnote{$w=2t$ using the terminology of \cite{ntru-prime}.} \\
\texttt{2.} $c_1(x)\leftarrow h(x)*r(x)\mod{(q,D(x))}$ \\
\texttt{3.} $ct(x)\leftarrow {\texttt{Round}}(c_1(x)) \ \ \ \ \# {\texttt{ the coefficients are multiple of 3}}$\\
%\texttt{4.} ${\textbf C}\leftarrow ({\textbf c},Hash({\textbf r},{\textbf h}))\hspace{14pt} \# {\texttt{ciphertext}}$\\
\texttt{4.}  ${\textbf K}\leftarrow Hash[{\texttt{encode}}(r(x))||{\texttt{encode}}(ct(x))]  $ \hspace{10pt} \# {\texttt{shared key}}\\
\texttt{5.}  {\textbf{return}} \ $({\textbf K},ct(x))$\\
\ \\
Remark that there is no need for ${\mathcal{L}}_m$ since the message is implicitly defined in line 3. Indeed, from remark \ref{tricky} there is a ternary polynomial $m(x)$ such that $c(x)=c_1(x)+m(x)=h(x)*r(x)+m(x)$ in ${\mathcal{R}}/q$ (exactly as in the previous flavors). We can rewrite line 3 as follows:\\\\
\texttt{3a.} $m(x)\leftarrow c_1(x)-{\texttt{Round}}(c_1(x))$\\
\texttt{3b.} $c(x)\leftarrow h(x)*r(x)+m(x)$ \text{in}\ ${\mathcal{R}}/q.$\\
In fact this is the way we implement {\texttt{encap}} function.\\
\ \\
The {\texttt{decap}} accepts three inputs, the secret key, the public key and the ciphertext. 
First computes the nonce $r(x)$ and then, it is easy, by following the code of encapsulation, to find the message $m(x)$ and the shared secret ${\bf K}.$\\\\
{\it Decapsulation}. Input the ciphertext $ct(x),$ the public $h(x)$ and the secret key $(f(x),g_3(x))$.
\ \\
\texttt{1.}  $e'(x)\leftarrow 3f(x)*ct(x)\mod{(q,D(x))}$ \\
\texttt{2.}  $e''(x)\leftarrow Lift_q(e'(x))$\\
\texttt{3.}  $e(x)\leftarrow e''(x)\mod{3}$\\
\texttt{4.}  $r''(x) \leftarrow e(x)*g_3(x) \mod{3}$ \\
\texttt{5.}  $r'(x) \leftarrow Lift_3(r''(x))$  \\
\texttt{6.} $c'(x) \leftarrow {\texttt{Round}}(h(x)*r'(x))$\\ 
\texttt{7.}${\textbf{If}}\ c'(x)=ct(x)\ {\textbf{then}}$ \\
\texttt{8.} \hspace{9pt} ${\textbf{return}}\ Hash[{\texttt{encode}}(r'(x))||{\texttt{encode}}(c'(x))] $
\section{A vfk lattice related to ntru lattice}\label{sec:Auxiliary}
	Let $k, N$ and $q$ be positive integers.	
We set 	
	$$M_k=
		\left[\begin{array}{c|c}
		I_N & -kI_N  \\
		\hline
		{\textbf 0}_N & qI_N   \\
		\end{array}\right].
		$$ 
Let ${\mathcal{L}}_{k}$ be the lattice generated by the rows of $M_k.$ 
Before we present a basic result (see the Proposition below) we remark that if $k=q$ the lattice is VFK. Indeed, we can check that the rows produce an obtuse superbasis,
\[{\textbf v}_i\cdot {\textbf v}_j= 0\ {\rm and}\ {\textbf v}_i\cdot {\textbf v}_{N+j}\le 0,\ (1\le i<j\le N),\]
and ${\textbf v}_0=-{\textbf v}_1-\cdots-{\textbf v}_{2N}=(-1,-1,...,-1;0,0,...,0),$ is such that 
$${\textbf v}_0\cdot {\textbf v}_j\le 0\ \ (j=1,2,...,2N).$$

\begin{proposition}
If $k\leq \sqrt{2q-1},$ then the lattice ${\mathcal{L}}_k$ is VFK.
\end{proposition}
\begin{proof}
We consider the unimodular matrix $U_P\in GL_{2N}({\mathbb{Z}})$ (for some integer $P$), 
\begin{equation}\label{matrix_change_basis}
U_P=\left[\begin{array}{c|c}
	I_N & {\textbf 0}_N  \\
	\hline
	PI_N & I_N   \\
\end{array}\right]=
\left[\begin{array}{cccccccc}
1 &     0 & \dots &   0      &  0   &  0   & \dots & 0  \\
0&     1 & \dots &   0       &   0&  0   & \dots & 0  \\
\vdots & \vdots  &  \ddots & \vdots &   \vdots & \vdots & \ddots & \vdots   \\
0 & 0 & \dots & 1            &   0    &   0  &  \dots &   0 \\
P & 0 &  \dots &    0        &     1     & 0     &  \dots &    0    \\
0 & P &  \dots &    0        &     0     & 1   &  \dots &    0    \\
\vdots & \vdots  &  \ddots & \vdots &   \vdots & \vdots & \ddots & \vdots   \\
0 & 0 &  \dots &    P        &     0     & 0   &  \dots &    1    \\
\end{array}\right].
\end{equation} 
		
We shall prove that we can find an integer $P$ such that, the rows of the matrix $U_PM_k$ is an obtuse superbasis of ${\mathcal{L}}.$ This integer $P$ will depend on $k$ and $q,$

We first compute the matrix
$$A = U_PM_k=$$
\begin{equation}\label{vfk-matrix}
 \left[\begin{array}{c|c}
I_N & -kI_N  \\
\hline
PI_N & (q-Pk)I_N   \\
\end{array}\right].
\end{equation}
The superbasis is 
$$	
\begin{array}{ccc}
{\textbf v}_1 & = &(1,0,...,0;-k,0,...,0)\\
{\textbf v}_2 & =  & (0,1,...,0; 0,-k,...,0)\\
... \\
{\textbf v}_N & = &(0,0,...,1;0,0,...,-k) \\
{\textbf v}_{N+1} & = &(P,0,...,0;q-Pk,0,...,0) \\
...\\
{\textbf v}_{2N} & = & (0,0,...,P;0,0,...,q-Pk)\\
\end{array}
$$ 
and
$${\textbf v}_0 = -\sum_{j=1}^{2N} {\textbf v}_j=(-1-P,...,-1-P;kP+k-q,...,kP+k-q).$$
Since ${\textbf v}_i$ and ${\textbf v}_{i+N}$ have non zero elements only at positions $i, N+i$ we conclude that
$${\textbf v}_i\cdot {\textbf v}_{j+N}=0,\ (i\not=j,\ 1\leq i,j\leq N).$$
Further,
\begin{equation}\label{vfk_inequality_1}
{\textbf v}_1\cdot {\textbf v}_{N+1} = \cdots={\textbf v}_i\cdot {\textbf v}_{i+N}=\cdots=(kP - q)k + P, \ (1\leq i \leq N).
\end{equation}			 
Similarly, 
$${\textbf v}_i\cdot {\textbf v}_j=0, \ \ {\textbf v}_{i+N}\cdot {\textbf v}_{j+N}=0,\ (1\leq i<j\leq N),$$
\begin{equation}\label{vfk_inequality_2}
{\textbf v}_0\cdot {\textbf v}_i =  -(kP + k - q)k - P - 1,\ (1\leq i\leq N),
\end{equation}			 
and
\begin{equation}\label{vfk_inequality_3}
{\textbf v}_0\cdot {\textbf v}_{i+N} =  (kP + k - q)(q-kP) - (P + 1)P,\ \ (1\leq i\leq N).
\end{equation}		
From (\ref{vfk_inequality_1}), in order to have ${\textbf v}_i\cdot {\textbf v}_{i+N}\leq 0$ \ $(1\le i\le n),$  we get 
$$P\le \beta=\frac{kq}{k^2+1}$$
and similarly from (\ref{vfk_inequality_2}) we get $$P\ge \alpha=\frac{kq-k^2-1}{k^2+1}.$$
Remark that,	
 $\beta - \alpha = 1.$ So, there is at least one   integer value in 
$[\alpha,\beta],$ the value 
\begin{equation}\label{value_of_P}
P=\lfloor kq/(k^2+1)\rfloor.
\end{equation}
Assume that $\beta\not\in {\mathbb{Z}}.$
Else, $P=\beta.$ 
Consider the polynomial, 
$$g(x) = -x^2(1+k^2)+x(-k^2+2kq-1)+q(k-q).$$
Therefore, it is enough the relation (\ref{vfk_inequality_3}) be non positive at $P$, i.e. $g(P)\leq 0.$
But the discriminant of $g(x)$ equals to $\Delta = (k^2+1)^2-4q^2.$ So, to have $g(x)\leq 0$ it is enough $\Delta\leq 0,$ therefore we conclude $(1+k^2)^2\le 4q^2.$ The result follows.
\end{proof}
\begin{remark}\label{remark:superbasis}
We set $r = - [(kP+k-q)k+P+1]$ and $$s=(kP+k-q)(q-kP)-(P+1)P.$$ 
The Selling parameters $q_{ij}={\textbf v}_i\cdot {\textbf v}_j,$ $0\le i,j\le 2N$ are given by the symmetric $(2N+1)\times (2N+1)$ matrix, $Q=(q_{ij})=$
$$=\left[\begin{array}{c|c|c}
N\big((1+P)^2+(k(P+1)-q)^2\big) & r{\textbf 1}_{1\times N}  & s{\textbf 1}_{1\times N} \\
\hline
r{\textbf 1}_{N\times 1}  & (1+k^2)I_N & \big((kP-q)k+P\big)I_N  \\
\hline
s{\textbf 1}_{N\times 1} &  \big((kP-q)k+P\big)I_N & \big(P^2+(q-Pk)^2\big)I_N   \\
\end{array}\right].
$$
where ${\textbf 1}_{1\times N}$ is a row with $N$ ones and ${\textbf 1}_{N\times 1}$ a column with $N$ ones.

\end{remark}
\begin{remark}
The maximum
%\footnote{see \url{https://github.com/drazioti/ntru_msg_recovey_attack/blob/main/auxiliary/find_max_k_and_P.sage}} 
$k$ such that $g(P)<0$ for various $q$

\begin{table}[h]

\begin{tabular}{|l||l|l|l|l|l|l||l|l||l|l|l|}
\hline 
$q$ & 32 & 64 & 128 & 256 & 512 & 1024 & 2048 & 4096 & 4621 & 4591 & 5167\\ \hline
$k$ & 8 & 11 & 16 & 23 & 32 & 47 & 64 & 91 & 101 & 98 &106\\ \hline
$P$ & 3 & 5 & 7 & 11 & 15 & 21 & 31 & 45 & 45 & 46 &48 \\  \hline
\end{tabular}
\caption{The values $q = 2048, 4096$ concern parameters for NTRU-HPS and the last three columns for the NTRU-Prime.}
\label{Tab:1}
\end{table} 

\end{remark}
\section{The Attack}\label{sec:attack}
Our method is adaptable for all three variations of NTRU. Across the three variants, there exists a ternary message $m(x),$ encrypted according to the formula: $c(x) = r(x)*h(x)+m(x)$ in ${\mathcal{R}}/q$, where $r(x)$ represents a nonce and $h(x)$ is the public key. In the case of NTRU-Prime, this message is defined by the equation $m(x) = {\texttt{Round}}(h(x) * r(x)) - h(x) * r(x)$ (see remark \ref{tricky}). 

In our method, we attack the message $m(x)$ of the system. So, if our attack is successful, we get the message. Now, having the message, we need to compute the nonce $r(x),$ since the shared key is a hash of the nonce and the ciphertext. It is enough to use the equation $r(x)=(c(x)-m(x))*h^{-1}(x)$ in ${\mathcal{R}}/q.$

Eve, who attacks the system, has the ciphertext $c(x)$ and the public key $h(x).$ Her purpose is to find the shared key ${\bf K}$. As we wrote, in ${\mathcal{R}}/q$ we have $$c(x)=r(x)*h(x)+m(x).$$
 Multiplying both sides with an integer $k$ (we shall choose $k$ later), we get, 
$$km(x)=b(x)+u(x),\text{ where \ } b(x)=k c(x)\ \text{\ and\ \ } u(x)=-kr(x)*h(x) \ \text{in}\ {\mathcal{R}}/q.$$
Therefore,  the previous equation in ${\mathcal{R}}$ is written
$$km(x)=b(x)+u(x)+qv(x),\text{ for\ some\ } v(x)\in {\mathcal{R}}.$$
Polynomials $m(x)$ and $u(x)$ are unknown. Let ${\bf m}=(m_i), {\bf b}=(b_i), {\bf u}=(u_i),$ and ${\bf v}$ be the vectors corresponding to $m(x), b(x), u(x),$ and $v(x),$ respectively.    We set ${\textbf V}$ to be the unknown vector $(-{\textbf m},{\textbf u}).$
We remark that 
$(-{\textbf m},{\textbf b}+{\textbf u})$ is in ${\mathcal{L}}_k,$
where ${\mathcal{L}}_k$ is the lattice generated by the rows of the matrix
\begin{equation}\label{vfk_lattice}
	M_k=
    \left[\begin{array}{c|c}
	I_N & -kI_N  \\
	\hline
	{\textbf 0}_N & qI_N   \\
	\end{array}\right].
\end{equation}  
Indeed, if we consider $(-{\bf m},-{\bf v})\in {\mathbb{Z}}^{2N},$ then
$$({-\textbf m},-{\textbf v})M_k=
(-{\textbf m},-{\textbf v})\left[\begin{array}{c|c}
			I_N & -kI_N  \\
			\hline
			{\textbf 0}_N & qI_N   \\
		\end{array}\right] = 	 (-{\textbf m},k{\textbf m} - q{\textbf v})=(-{\textbf m},{\textbf b}+{\textbf u}).$$
  Now, we want a nice approximation of the unknown vector ${\bf V}.$
Assume that we can find a vector ${\textbf E}=(E_i)\in {\mathbb{Z}}^{2N}$ such that,
\begin{equation}\label{rel:1}
||{\textbf V}-{\textbf E}||<\frac{1}{2}\lambda_1,\ {\rm where\ } \lambda_1\ {\rm is\ the\ length\ of\ a\ shortest\ vector\ in\ } {\mathcal{L}}_k.
\end{equation}
  Note that, neither ${\bf V}$ nor ${\bf E}$ is in ${\mathcal{L}}_k.$ Further, we remark that the first successive minima $\lambda_1$ can be computed in polynomial time in ${\mathcal{L}}_k$ since it is VFK \cite{VFK-SVP}.
We choose the target vector ${\bf t}$ through ${\bf E}$ as follows,
$${\textbf t}=(E_1,...,E_N,b_1+E_{N+1},...,b_N+E_{2N}),$$
and set ${\textbf W}\leftarrow CVP({\mathcal{L}}_k,{\textbf t}).$ We shall prove that ${\bf W}$ provides the message ${\textbf m}.$ First, we remark that
$$||{\textbf W} - {\textbf t}||\le ||(-{\textbf m},{\textbf b}+{\textbf u})-{\textbf t}||.$$ Then,
$$\|(-{\bf m},{\bf b}+{\bf u})-{\bf t})\| =$$
$$=\|(-m_1,...,-m_N,b_1+u_1,...,b_N+u_N)-(E_1,...,E_N,E_{N+1}+b_1,...,E_{2N}+b_{N})\|=$$
$$=\|(-m_1-E_1,...,-m_N-E_N,u_1-E_{N+1},...,u_N-E_{2N})\|=$$
$$=\|(-{\bf m},{\bf u})-{\bf E}||=||{\textbf V}-{\textbf E}\|<\frac{1}{2}\lambda_1.$$
\begin{remark}\label{remark:svp}
We remark here that, since ${\mathcal{L}}_k$ is VFK lattice, we can also compute $\lambda_1$ in polynomial time (see \cite{VFK-SVP}). We have implemented this algorithm in sagemath, see appendix \ref{appendix:svp-vfk}. By inspection we note that $\lambda_1=\sqrt{1+k^2}.$
\end{remark}
Finally, 
$$\|{\textbf W}-(-{\textbf m},{\textbf b}+{\textbf u})\|=\|({\bf W}-{\bf t}) + \big({\bf t}-(-{\bf m},{\bf b}+{\bf u})\big)\|$$
$$\le \|{\textbf W}-{\textbf t}\| +  \|{\bf t} -(-{\textbf m},{\textbf b}+{\textbf u})\|\le 2\|(-{\textbf m},{\textbf b}+{\textbf u}) - {\textbf t}\|<\lambda_1.$$
But ${\textbf W}-(-{\textbf m},{\textbf b}+{\textbf u})\in {\mathcal{L}}_k$ thus, ${\textbf W}=(-{\textbf m},{\textbf b}+{\textbf u})$. We conclude therefore that the first $N-$coordinates of ${\textbf W}$ provide the (minus) message ${\textbf m}.$ Having the message ${\bf m},$ i.e., the ternary polynomial $m(x),$ we can find the nonce $r(x)$ since $c(x)=r(x)*h(x)+m(x).$ Then, computing the hash of 
${\texttt{encode}}({r(x)})||{\texttt{encode}}({c(x)}),$ Eve calculates the shared secret ${\bf K}.$

To summarize: \\
$({\rm{\bf{i}}}).$ Find a positive integer $k$ and $P$ such that: the lattice ${\mathcal{L}}_k$ generated by the rows of the matrix (\ref{vfk_lattice}) is VFK. $P$ is chosen as in (\ref{value_of_P}). This $k$ depends on the parameters of the NTRU, especially $q.$  For instance, if $q=1024$ we choose $k=47$ (from Table \ref{Tab:1}),\\
$({\rm{\bf{ii}}}).$ Choose ${\bf E}=(E_i)\in {\mathbb{Z}}^{2N}$ in order to construct the target vector, 
$${\textbf t}=(E_1,...,E_N,b_1+E_{N+1},...,b_N+E_{2N})$$
where ${\bf b}=(b_i)$ corresponds to polynomial $b(x)=kc(x),$ and $c(x)$ is the ciphertext. I.e. ${\bf b}$ is a known vector. In our attack we chose $E_1=E_2=\dots=E_N=0.$ Since ${\bf 0}_N$ is a {\it{good}} approximation of ${\bf m}\in\{-1,0,1\}^N.$ So, ${\bf E}=({\bf 0}_N,E_{N+1},...,E_{2N})$ and 
$${\textbf t}=(0,...,0,b_1+E_{N+1},...,b_N+E_{2N}),$$
$({\rm{\bf{iii}}}).$ Call the VFK variant of CVP for the pair $({\mathcal{L}}_k,{\bf t}).$ In order to execute CVP for VFK lattices, we also need the Selling parameters $Q$ (given in remark \ref{remark:superbasis}),\\
$({\rm{\bf{iv}}}).$ Compute the possible message $m(x),$\\
$({\rm{\bf{v}}}).$ Compute the possible shared secret ${\bf K}$,\\
$({\rm{\bf{vi}}}).$ Verify that  ${\bf K}$ is the shared secret.\\\\
To model our attack in step $({\rm{\bf{ii}}})$, we use a range $R,$ which is a positive integer. We assume the existence of an oracle such that on input the public key and the range $R$, outputs $E_1'=E_{N+1},...,E_{N}'=E_{2N}$: $|E_{i}'-u_i|\le R$ for $i=1,...,N.$ In other words, the polynomial 
$$u(x)-\sum_{i=1}^{N}E_{i}'x^{i-1}=-k(h(x)*r(x))-\sum_{i=1}^{N}E_{i}'x^{i-1}$$
has all its coefficients in $\{-R,...,R\}.$
The attack may fail if our guessing for vector ${\bf E}'=(E_1',...,E_N')$ is not good, or, in other words, if $R$ is large enough. Through our experiments, we aim to determine an upper bound, denoted as $R_0,$ beyond which our attack is no longer effective. Each call to the oracle with a different seed provides a new random set 
$\{E_i'\}$ with the same property. We assume that the distribution of $E_i'$ with this property is uniform.

To implement step $({\rm{\bf{vi}}})$, we start some encrypted conversation with the owner of the secret key, by using a symmetric cryptosystem and so indirectly, we can check the validity of the shared key.
\begin{remark}
In NTRU-HPS, where we use messages with weight $q/8-2,$ we get
$$\|{\bf V}-{\bf E}\|^2=\sum_{i=1}^Nm_i^2 + \sum_{i=1}^N(u_i-E_i')^2<\frac{q}{8}-2+NR^2.$$
So, if $$\frac{q}{8}-2+NR^2<\frac{\lambda_1^2}{4},\ \text{i.e.} \ R^2<\frac{1}{N}\Big(\frac{\lambda_1^2}{8}-\frac{q}{8} + 2\Big),$$
the attack will succeed. Furthermore, if we use remark \ref{remark:svp}, and we set $\lambda_1^2=1+k^2,$ we get
$$R^2<\frac{1}{N}\Big(\frac{1+k^2}{8}-\frac{q}{8} + 2\Big).$$
\end{remark}
\subsection{More remarks about our attack}
In our approach, we assume that we do not have any information about the message ${\bf m}.$ However, we do presume that we have some degree of information about (a multiple of) the nonce $r(x).$ 

The challenge in our attack lies in approximating two key unknowns: the message itself and the specific multiple of the nonce, denoted as $-kh(x)*r(x)$. 
This combined unknown vector is represented as ${\bf V}=(-{\bf m},{\bf u})$, where ${\bf m}$ signifies the message and ${\bf u}$ corresponds to $-kh(x)*r(x)$.
We approximate the message with the zero vector (since $m(x)$ is a ternary polynomial) and to quantify how much information we know about the nonce, we use the range $R.$ Now, as $R$ increases, the amount of information we obtain about the nonce diminishes significantly, but the information we have about the message remains zero. In practice, obtaining such an oracle can be a challenging task. Nevertheless, it's essential to acknowledge that side channel attacks, as exemplified in \cite{kamal}, play a significant role in the assessment of Post-Quantum Cryptosystems like NTRU. Therefore, it is conceivable to construct the specific oracle through a side channel attack.

\section{Experiments}\label{sec:experiments}
We execute experiments for NTRU-HPS and NTRU-Prime, where we implemented encryption-decryption (or encap-decap) algorithms. First, we choose $k, P$, so that ${\mathcal{L}}_k$ is a VFK lattice. The oracle accepts the public key, a seed, and the range $R,$ and outputs $N$ integers $E_i',$ such that $|u_i-E_i'|\le R \ (1\le i\le N)$. Finally, we define the target vector ${\bf t}=({\bf 0}_N,E_1',...,E_N')$ and the VFK-CVP algorithm that is executed for the pair $({\mathcal{L}}_k,{\bf t}).$ For each message, we make $100$ calls to the oracle and if we manage to find the message ${\bf m}$, we consider our attack successful. 

In the following subsection, we update the experimental results of \cite{adam_draz} and in subsection \ref{subsec:ntru-prime} we provide some new results concerning NTRU-Prime.
\subsection{NTRU-HPS: Comparison with results of \cite{adam_draz}}\label{subsec:ntru-hps}
In the example 1 of \cite{adam_draz}, $(N,p,q,d)=(239,3,2^{8},71)$ are used and the attack provides range $R_0=9$. With our enhanced method,
%\footnote{see \href{https://github.com/drazioti/ntru\_msg\_recovey\_attack/tree/main/ntru-hps/comparison\_with\_older\_results}{here.}}, 
using a CVP oracle, instead of an approximation CVP oracle, we get $R_0=11.$

Using Babai's approximation algorithm for CVP on the previous VFK lattice, we get similar results. This indicates a notable consistency in outcomes, even with the distinct methodologies employed in each scenario.

Additionally, when we apply a VFK-style matrix, like the one used in CVP, along with Babai's approximation algorithm, we achieve similar results. It appears that Babai's algorithm works particularly effectively in these specific VFK lattices. %	See also Fig.\ref{Fig.3}, which tell us for the previous parameters the closest vector and the output of Babai's algorithm is very close.

For our next experiment, we utilized the parameters $(N,p,q,d)=(509,3,2048,10)$. This example, also mentioned in \cite[example 5]{adam_draz}, is especially noteworthy as it pertains to the {\texttt{ntruhps2048509}}. According to \cite{adam_draz}, the maximum range $R=R_0$ for this case is limited to $26.$ 
Our new attack, that utilizes the min-cut algorithm, has considerably improved this limit to $R_0=32$. It is worth noting that using the min-cut algorithm has been an improvement over the Babai algorithm which achieves a range of $R_0=31$, which is not a trivial difference. To see this, consider two $N$ dimensional cubes both centered at ${\bf u}$. One has a side length of $2R_0$, and the other has a side length of $2(R_0+1)$. In ${\mathbb{Z}}^N$, the difference in the number of points inside these cubes is $(2R_0+3)^N - (2R_0+1)^N$, showing a significant exponential increase in the number of points.

For $(N,p,q,d)=(677,3,2048,20)$ i.e. concerns {\texttt{ntruhps2048677}}, the experiment in \cite[example 6]{adam_draz} provides $R_0=17$ and our new attack provides $R_0=32.$ 

{\bf Some new results}. We also executed the attack to {\texttt{ntruhps4096821}} where we found $R_0=45.$ In this experiment, it's important to note that we strictly adhere to the variant of NTRU-HPS submitted to NIST. While our previous examples (in the previous paragraph) employed older parameter sets and unweighted messages, these differences don't appear to significantly impact the success of our attack.
\subsection{NTRU-prime}\label{subsec:ntru-prime}
For NTRU-prime
%\footnote{For NTRU-prime implementation we used descent code from \url{https://ntruprime.cr.yp.to/sntrup4591761.sage}, but without encoding-decoding functions. For details see \url{https://github.com/drazioti/ntru_msg_recovey_attack}} 
we tried the following parameters : 
\begin{center}
{\texttt{sntrup653}}, {\texttt{sntrup761}} and {\texttt{sntrup867}}.  
\end{center}
The flavors {\texttt{sntrup653}}, {\texttt{sntrup761}} and {\texttt{sntrup867}}, have $(p,q,w)=$ $(653,4621,288),$ $(761,4591,286),$ and $(857,5167,322),$ respectively. As previously, with the term {\it{attack}}, we mean, $100$ calls to the oracle (for a fixed message). Our attack revealed the messages for {\texttt{sntrup653}} with maximum range $R_0=50$; for {\texttt{sntrup761}}, it's $48$ and for {\texttt{sntrup867}} it's $53.$ Each call to the CVP-oracle for VFK lattices took on average 4-5 minutes (wall time), on a modern PC.

One might anticipate that {\texttt{sntrup653}} and {\texttt{ntruhps2048677}} would yield similar outcomes due to their comparable dimensions. However, it appears that NTRU-HPS displays greater resilience against our attack, since, $R_0({\rm{HPS}})=32$ is much smaller than to $R_{0}({\rm{Prime}})=50.$ However, given the substantial differences between these two variants, definitive conclusions about their comparative resilience require further study.

\begin{remark} In our investigation, we successfully executed attacks on NTRU-HPS and NTRU-Prime however, we did not extend this analysis to NTRU-HRSS. Our rationale behind this decision lies in our expectation of consistent outcomes similar to the previous versions. Upon closer examination of the attack methodology and its application, we discern that its efficacy isn't contingent on the specific nature of the underline ring. Rather, the critical factor for success hinges on acquiring information about the approximation of the vector ${\bf b} + {\bf u}$. Leveraging an Oracle allows us to obtain this information, somewhat circumventing the ring's inherent properties. Consequently, the NTRU problem transforms into a lattice problem, primarily tied to the lattice's basis rather than the ring's characteristics.
\end{remark}

\section{Conclusion}\label{sec:conclusions}
In this work, we used lattice theory to deliver a message recovery attack in all the NTRU variants. Multiplying the encryption equation with a suitably chosen constant we create an equivalent equation, where an unknown value, containing the message, it belongs to a VFK lattice. This allows us to use polynomial time algorithm for SVP and CVP, where we then manage to exploit a message recovery attack with the help of an oracle. Then, a target vector is carefully selected, and the CVP oracle is invoked. If some conditions are fulfilled we get the message, which allows to find the shared key.  Usually, an attack targets the key of the cryptosystem, but here we execute a message recovery attack, which are rare in the bibliography. Additionally,  the connection between NTRU and VFK is new and hasn't been explored in the bibliography that exists today.

 The hard part is finding a vector ${\bf E}$, that's an approximation of the one we don't know. To do this, we need help from a specific kind of oracle. This aspect presents a limitation to the attack as obtaining such an oracle in practical scenarios proves to be challenging. To mitigate this, we can employ a side-channel attack.
Several studies, such as \cite{aydin,kamal,vizev}, have explored side-channel attacks in this context. Especially in \cite[chapter III(A)]{sidechannel}, the authors suggest a power-based side-channel attack that targets the random generation of polynomials in NTRU.
By exploiting a plaintext message attack and collecting numerous pairs $(msg,ct)$, we can identify any potential biases in the distribution of the nonce $r(x)$ computed by $r(x)=h^{-1}(x)(msg(x)-ct(x)),$ if such biases exist.

Finally, we illustrated the success of the attack through several examples and proposed how a value $R_0$ could possibly be used to indicate the potential resistance of NTRU to the presented attack.

\appendix
\newpage
\section{Pseudocode for CVP in VFK lattices}\label{appendixA}

%% This is needed if you want to add comments in your algorithm with \Comment
\RestyleAlgo{ruled}

\SetKwComment{Comment}{/* }{ */}
\begin{algorithm}[H]
\caption{Solving CVP in VFK lattices}\label{alg:MAINALGO}
% \KwData{$n \geq 0$}
% \KwResult{$y = x^n$}
\SetKwInOut{Input}{input}
\SetKwInOut{Output}{output}
\Input{$B \in \mathbb{N}^{^{{n+1}\times m}} $ matrix with superbasis vectors $\{\textbf{b}_i\}, i \in \{1,...,n+1\}$ as rows and $\textbf{y} \in \mathbb{R}^{1 \times m} $ the target vector}
\Output{A closest vector $\textbf{v} \in \L$ of $\textbf{y}$, where $\L$ the lattice generated by the rows of $B$}
$q_{ij} \gets \textbf{b}_i\cdot \textbf{b}_j$ \Comment*[r]{for $i,j = 1, ..., n+1, i \neq j $}
${\bf q} \gets (q_{ij})$ \ ;\\
find ${\bf z}\in {\mathbb{R}}^{1\times (n+1)}$ such that $\textbf{y} = {\bf z}B$ \ ;\\ 
${\bf u} \gets \lfloor {\bf z} \rfloor$\;
\For{$i=1$, \textbf{to} $n+1$} {
  $G\gets flow\_network({\bf z},{\bf u}, {\bf q})$ \Comment*[r]{see alg. 3}
  ${\bf t}\gets mincut(G)$ \Comment*[r]{see alg. 4}
  ${\bf u}\gets {\bf u}+{\bf t}$\;
}
\end{algorithm}

\SetKwComment{Comment}{/* }{ */}

Now we provide the pseudocode of flow network.\ \\

\begin{algorithm}[h!tb]
\caption{Flow Network}\label{alg:GRAPH}
\SetKwInOut{Input}{input}
\SetKwInOut{Output}{output}
\Input{target vector $\textbf{z} \in \mathbb{R}^{n+1 \times 1}$, $\textbf{u} \in \mathbb{N}^{n+1 \times 1} $, selling parameters $q_{ij}$}
\Output{An undirected flow network $G$}
${\bf p} \gets {\bf z} - {\bf u}$\;
Create Graph $G$ with $n+3$ nodes and edges connecting them\;
source $\gets$ first node\;
sink $\gets$ last node\;
Edges $e_{ij}$ have the following weights: \\
\For{$i=1$, \textbf{to} $n+1$}{
\For{$j=i+1$, \textbf{to} $n+1$}{
$w_{ij} = - q_{ij}$\;
}
$s_i = -2 \sum_{j=1}^{n+1} q_{ij} p_j$\;
\eIf{$s_i \geq 0$}{
    $w_{i,sink} \gets s_i$\;
    $w_{source,i} \gets 0$\;
}{
    $w_{i,sink} \gets 0$ \;
    $w_{source,i} \gets -s_i$\;
}
}
\Return{$G$}
\end{algorithm}

The minimum cut algorithm of \ref{alg:MINCUT}, uses the Ford-Fulkerson procedure \cite{MAXFLOW} for undirected graphs, that finds the max-flow of a graph, which is equivalent to the minimum cut \cite{MAXFLOWMINCUT}. $f(i,j)$, for $i,j \in E(G)$ describes the flow from the node $i$ to node $j$. In our experiments we are using the implementation of networkx's mincut algorithm\footnote{see \url{https://networkx.org/}}.\ \\

\SetKwComment{Comment}{/* }{ */}
\begin{algorithm}[H]
\caption{MinCut}\label{alg:MINCUT}
\SetKwInOut{Input}{input}
\SetKwInOut{Output}{output}
\Input{Graph $G$, source node $s$, sink node $t$}
\Output{vector $\textbf{t} \in \{0,1\}^{n+1}$}

\For{each edge $(u,v) \in E(G)$}{
 $f(u,v) \gets 0$\;
 $f(v,u) \gets 0$\;
}
\While{ $\exists$ a path $p$ from $s$ to $t$ with no cycles in the residual network $G_f$}{
    $c_f(p) \gets \min\{ c_f(u,v) : (u,v) \in p \}$ \Comment*[r]{residual capacity}
    \For{ each edge $(u,v) \in p$}{
    $f(u,v) \gets f(u,v) + c_f(p)$\;
    $f(v,u) \gets - f(u,v)$\;
    }
}
$C \gets$ all the vertices that exist on the max-flow path\;
\For{$i=1$, \textbf{to} $n+1$}{
\eIf{$v_i \in C$}{
    $t_i=1$\;
}{
    $t_i=0$\;
}
}
\Return{\textbf{t}}
\end{algorithm}

\section{Babai vs CVP on VFK}\label{appendixB}
\begin{figure}[h!tb]
\centering
\rotatebox{0}{\scalebox{0.5}{\includegraphics{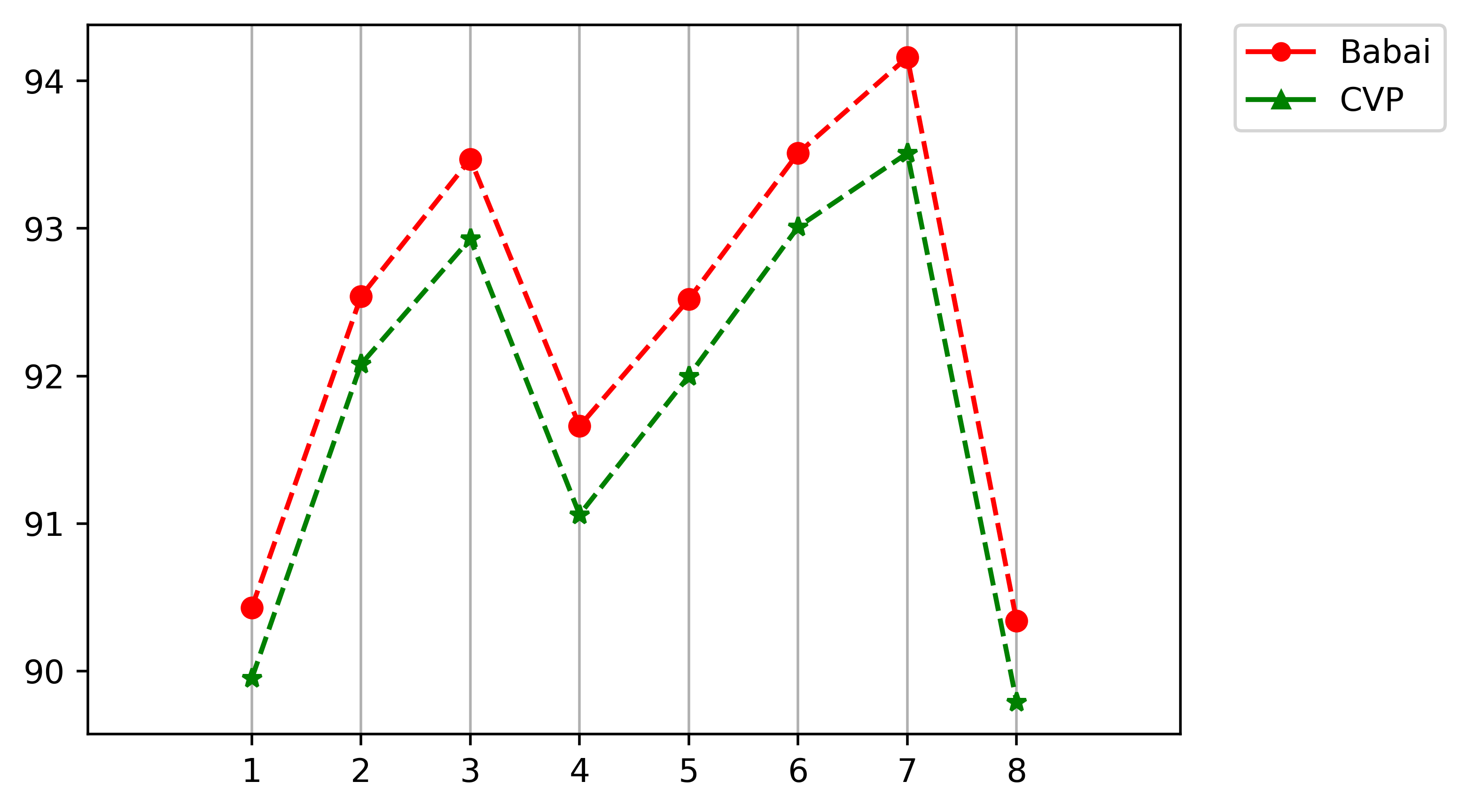}}}
\caption{We use the matrix (\ref{vfk-matrix}) with $N=677$ and parameters $(P,q,k)=(6,70,10).$ We execute Babai and CVP (using the min-cut algorithm) with various target vectors of the form ${\textbf t}\xleftarrow{\$} \{0,1\}^N \times \{-1000,...,1000\}^N.$ The $y-$axis contains $||{\textbf x}-{\textbf t}||,$ where ${\textbf x}$ is either the output from CVP or Babai.}
\label{Fig.1}
\end{figure}

\begin{figure}[h!tb]
\centering
\rotatebox{0}{\scalebox{0.5}{\includegraphics{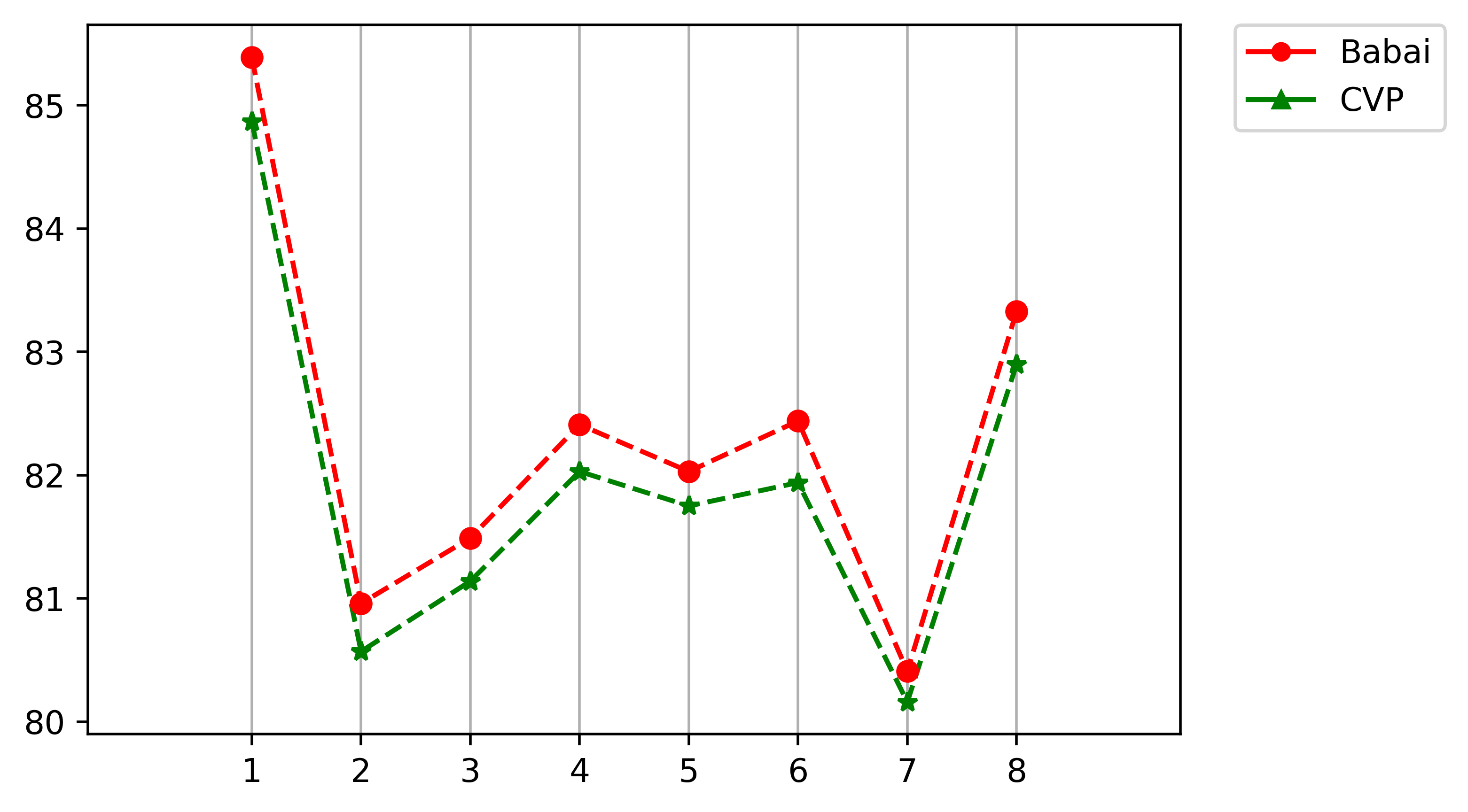}}}
\caption{We fix the matrix (\ref{vfk-matrix}) with $N=557$ and parameters $(P,q,k)=(6,70,10).$  We execute Babai and CVP (using the min-cut algorithm) with various  target vectors of the form ${\textbf t}\xleftarrow{\$} \{0,1\}^N \times \{-1000,...,1000\}^N.$ The $y-$axis contains $||{\textbf x}-{\textbf t}||,$ where ${\textbf x}$ is either the output from CVP or Babai. The $x-$axis counts the number of experiments.}
\label{Fig.2}
\end{figure}
\section{}\label{appendix:svp-vfk}

\begin{table}[htb!]

\begin{tabular}{|l||l|l|l|l|}
\hline 
variant & ntruhps2048509 & ntruhps2048677 & sntrup653 & sntrup761 \\ \hline
$N$ & 509 & 677 & 653 & 761\\ \hline
$q$ & 2048 & 2048 & 4621 & 4591\\ \hline
$k$ & 64 & 64 & 101 & 98  \\  \hline
$P$ & 31 & 31 & 45 & 46 \\  \hline\hline
$\lambda_1$ & 64.0078... & 64.0078.. & 101.004 & 98.005...\\ \hline
\end{tabular}
\caption{We computed the first successive minima $\lambda_1$ for the VFK lattices ${\mathcal{L}}_k$. Note that in all the examples $\lambda_1=\sqrt{1+k^2}.$ The code is in the github repository. }
\label{Tab:2}
\end{table} 
\end{document}